\documentclass[10pt]{article}
\usepackage{amsmath, amssymb,amsthm}
\usepackage{color}
\usepackage{dsfont}
\usepackage{enumerate}

\numberwithin{equation}{section}
\newcommand{\Int}{\displaystyle\int\limits}
\newtheorem{theo}{{\bf{Theorem}}}[section]
\newtheorem{prop}[theo]{{\bf Proposition}}
\newtheorem{lem}[theo]{{\bf Lemma}}
\newtheorem{Definition}[theo]{\bf Definition}

\newtheorem{rem}{{\bf Remark}}[section]

\newtheorem{remark}[rem]{Remark}
\newtheorem{defi}{{\bf Definition}}[section]

\renewcommand{\proof}{\noindent{\bf Proof.\ }}

\newcommand{\no}{\nonumber}
\newcommand{\noi}{\noindent}

\newcommand{\pa}{\partial}

\newcommand{\RR}{\mathbb{R}}
\newcommand{\si}{\sigma}

\newcommand{\la}{\lambda}

\newcommand{\sig}{\sigma}
\newcommand{\tab}{\hspace*{0.3in}}

\newcommand{\tS}{\tilde{S}}
\newcommand{\vf}{\varphi}
\newcommand{\vs}{\varsigma}

\newcommand{\vp}{\varepsilon}

\catcode`\@=11
\catcode`\@=12

\begin{document}
	
	\title{Pricing Derivatives in a Regime Switching Market with Time Inhomogeneous Volatility\thanks{The research of the first author was supported by UGC Fellowship.}}
	\author{Milan Kumar Das\thanks{IISER, Pune 411008, India; email: milankumar.das@students.iiserpune.ac.in.}\qquad\qquad Anindya Goswami\thanks{IISER, Pune 411008, India;
			email: anindya@iiserpune.ac.in.}\qquad\qquad Tanmay S. Patankar\thanks{IISER, Pune 411008 India; email: tanmaysp1004@gmail.com.} }
	
	\date{}
	
	\maketitle
	
	{\bf Abstract:} This paper studies pricing derivatives in an age-dependent semi-Markov modulated market. We consider a financial market where the asset price dynamics follow a regime switching geometric Brownian motion model in which the coefficients depend on finitely many age-dependent semi-Markov processes. We further allow the volatility coefficient to depend on time explicitly. Under these market assumptions, we study locally risk minimizing pricing of a class of European options. It is shown that the price function can be obtained by solving a non-local B-S-M type PDE. We establish existence and uniqueness of a classical solution of the Cauchy problem. We also find another characterization of price function via a system of Volterra integral equation of second kind. This alternative representation leads to computationally efficient methods for finding price and hedging. Finally we analyse the PDE to establish continuous dependence of the solution on the instantaneous transition rates of semi-Markov processes. An explicit expression of quadratic residual risk is also obtained.

	
	{\bf Classification No:} 60K15, 91B30, 91G20, 91G60.
	\section{Introduction}
In 1971 Black, Scholes and Merton considered a mathematical model of asset price dynamics to find an expression of price of a European option on the underlying asset. In their model, the stock price process is modeled with a geometric Brownian motion. The drift and the volatility coefficients of the price are taken as constants. Since then, numerous different improvements of their theoretical model are being studied. Regime switching models are one such extension of the Black-Scholes-Merton(B-S-M) model. In regime switching model it is assumed that the market has finitely many hypothetical observable  possible economic states and those are realized for certain random intervals of time. The key market parameters are assumed to depend on those regimes or states and the state transitions are modeled by a pure jump process. Extensive research has been done to study markets with Markov-modulated regime switching \cite{BAS},\cite{BE},\cite{DES},\cite{DKR},\cite{ECS},\cite{AGK},\cite{GZ},\cite{RR},\cite{MR}.
There are also some further generalization, carried out by several authors by introducing jump discontinuities in the asset dynamics along with Markov regimes. In all these works the possibility of switching regimes is restricted to the class of finite state Markov Chains.

	In comparison with Markov switching, the study of semi-Markov modulated regime switching is relatively uncommon. In this type of models one has opportunity to incorporate some memory effect of the market. In particular, the knowledge of past stagnancy period can be fed into the option price formula to obtain the price value. Hence this type of models have greater appeal in terms of applicability than the one with Markov switching. The pricing problem with semi-Markov regimes was first correctly solved in \cite{AGMKG}.  It is important to note that the regime switching models lead to incomplete markets. Since there might be multiple no arbitrage prices of a single option, one needs to fix an appropriate notion to obtain an acceptable price. Option pricing with a special type of semi-Markov regime is studied in \cite{AGMKG} using F\"{o}llmer Schweizer decomposition \cite{FS}. There it is shown that the price function satisfies a non-local system of degenerate parabolic PDE. In a recent paper \cite{AJP} the same problem for a more general class of age-dependent processes is studied. An age-dependent process $\{X_t\}_{t\geq 0}$ on  $\mathcal{X}:=\{1,\ldots,k\}\subset \mathbb{R}$ is specified by its instantaneous  transition rate function  $\la:\{(i,j)\in \mathcal{X}^2|i\neq j\}\times [0,\infty)\rightarrow (0,\infty)$  and is defined by the strong solution of the following system of stochastic integral equations
		\begin{eqnarray}
		X_t &=& X_0 + \int_{(0,t]}\int_{\mathbb{R}} h_\lambda(X_{u-}, Y_{u-},z)\wp(du,dz)\text{\label{1}}\\
	Y_t &=& t- \int_{(0,t]} \int_{\mathbb{R}} g_\lambda(X_{u-}, Y_{u-},z) \wp(du,dz),\text{\label{2}}
	\end{eqnarray}
	where $\wp(du,dz)$ is a Poisson random measure with intensity $dudz$, independent of $X_0$ and
	$$  h_\lambda(i,y,z) := \sum_{j \in {\mathcal{X}\setminus \{i\}}} (j-i) 1_{\Lambda_{ij}(y)}(z), ~~~  g_\lambda(i,y,z) := \sum_{j \in {\mathcal{X}\setminus \{i\}}} y 1_{\Lambda_{ij}(y)}(z),$$
	where for each $y\geq 0$, and $i\neq j$, $\Lambda_{ij}(y)$ are the consecutive (with respect to the lexicographic ordering on $\mathcal{X}\times \mathcal{X}$) left closed and right open intervals of the real line,
	each having length $\lambda_{ij} (y)$. We clarify that if $\{(X_t,Y_t)\}_{t\geq 0}$ is the solution of \eqref{1}-\eqref{2}, then $X_t$ is called the age-dependent process and $Y_t$ is called the age process. It is shown in (Th. 2.1.3, \cite{tanmay}) that an age dependent process is a semi-Markov process.
	
	In both the papers \cite {AGMKG} and \cite{AJP}, all the market parameters, namely spot interest rate $r$, drift coefficient $\mu$ and volatility coefficient $\si$ depend on a single semi-Markov process. We recall that although the joint process of two independent Markov processes is Markov, the same phenomena is not valid for semi-Markov case. For this reason, assumption of a single semi-Markov process to derive both $r$ and $\si$ is rather restrictive.  To overcome this restriction, in this paper, we consider a componentwise semi-Markov process (CSM), which is a wider class of pure jump processes than those in \cite {AGMKG} and \cite{AJP}.  A pure jump process $X$ on a finite state space $\mathcal{S}$ is called a CSM if there is a bijection $\Gamma:\mathcal{S}\to \mathcal{X}^{n+1}$ for some non-empty finite set $\mathcal{X}$, and some non-negative integer $n$ such that each component of $\Gamma(X)$ is semi-Markov process, independent to each other. To model the regimes of the market, we consider a CSM $\{X_t\}_{t\geq 0}$ on $\mathcal{X}^{n+1},$ where $X^l$, the $l^{\text{th}}$ component of $X$, is an age-dependent process with instantaneous rate functions $\la^l$, for every $l=0,\ldots,n$. We denote the age process of $X^l$ as $Y^l$ and $Y$ defined as $(Y^0,\ldots,Y^n)$ is the age process of $X$.
	
	In many regime switching models of asset price dynamics, the volatility coefficients do not posses explicit time dependence (see \cite{BAS},\cite{BE},\cite{DES},\cite{DKR},\cite{ECS},\cite {AGMKG},\cite{AGK},\cite{GZ},\cite{RR},\cite{MR}). In such time homogeneous models the volatility $\si$ can take values from a finite set only. Such models fail to capture many other stylized facts including periodicity feature of $\si$. In the present model, we allow $\si$ to be time inhomogeneous.
	
	 In this paper, we consider a market with one locally risk free asset with price $S^0$, and $n$ risky assets with prices $\{S^l\}_{l=1,\ldots,n}$, and  address locally risk-minimizing pricing for a contingent claim $K(S_T)$. Here we consider a wide range of functions $K:\mathbb{R}^n_+\to \mathbb{R}_+$, which includes vanilla basket options. We show that the price of the claim at time $t$, when $(S^l_t,X^l_t,Y^l_t)$ is $(s^l,x^l,y^l)$, for each $l$, is a function $\vf$ of $(t,s=(s^1,s^2,\ldots,s^n),x=(x^0,x^1,\ldots,x^n),y=(y^0,y^1,\ldots,y^n))$ and that satisfies a Cauchy problem.
	In order to write the equation we use a notation $ R^l_j v $, for a vector $ v\in \mathbb{R}^{n+1} $ to denote the vector $ v+(j-v^l)e_l $, in which the $ l^{\textrm{th}} $ component of $ v $ is replaced with $ j $. The system of PDE is given by
	\begin{align}\label{mainpde}
	\no &\frac{\partial \vf}{\partial t} (t, s, x, y)+ \sum_{l=0}^n \frac{\pa \vf}{\pa y^l}(t, s, x, y) +r(x)\sum_{l=1}^n s^l \frac{\pa \vf}{\pa s^l}(t, s, x, y)+\frac{1}{2}\sum_{l=1}^{n}\sum_{l'=1}^{n}a^{ll'}(t,x)s^ls^{l'}\frac{\partial^2\vf}{\partial s^{l}\pa s^{l'}}(t,s,x,y)  \\
	&+ \sum_{l=0}^n \sum_{j\neq x^l}\lambda_{x^lj}^l(y^l)\left[ \vf(t,s,R^l_j x,R^l_0 y)-\vf(t,s,x,y) \right] = r( x) ~ \varphi(t, s, x, y),
	\end{align}
	\noi defined on
	\begin{equation*}
	\mathcal{D}:= \{ (t, s, x, y)\in (0,T)\times (0,\infty)^n\times
	\mathcal{X}^{n+1}\times (0,T)^{n+1} \mid y \in (0,t)^{n+1}\},
	\end{equation*}
	\noi and with conditions
	\begin{align}
	\varphi (T, s, x, y)=& K(s); \tab s\in [0,\infty)^n; \tab 0 \le
	y^l\le T ; \tab x^l\in\mathcal{X},\tab l=0,1,\dots,n, \label{pdeboundary}
	\end{align}
	where the diffusion coefficient $a:=(a^{ll'})_{n\times n}$ is continuous in $t$.

We note that \eqref{mainpde} is a linear, parabolic, degenerate and non-local PDE. The non-locality is due to the occurrence of the term $\vf(t,s,R^l_j x,R^l_0 y)$, where $R^l_0 y$ need not be same as $y$ in general. We establish existence and uniqueness of the classical solution in this paper. We also find a Volterra integral equation of second kind, which is equivalent to the PDE. Using the Banach fixed point Theorem, we show the integral equation has a unique solution. Thus we show that one can find the price function by solving the integral equation which is computationally more convenient than solving the PDE. We also obtain an expression of optimal hedging involving integration of price function.

  This observation essentially leads to a robust computation of optimal hedging. Finally we carry out a sensitivity analysis to establish continuous dependence of solution of the PDE \eqref{mainpde}-\eqref{pdeboundary}	on transition rate functions. This result assures close approximation of price when instantaneous rate is approximated by a consistent estimator as in \cite{AGSN}.
	
 The rest of this paper is arranged in the following manner. We present model description in Section $2$. In this section we first study a class of componentwise semi-Markov processes and then we describe the the asset price dynamics. We have also shown that under admissible strategies the market is arbitrage free.  Section 3 presents the approach of option pricing. In this section we state the main result of the paper. In Section $4$, we establish the existence, uniqueness and regularity of solution of a Volterra integral equation which is shown to be equivalent to the PDE in the next section. Section $5$ deals with the well-posedness of the PDE. In this section we also derive certain properties of the solution and its derivative. Using the results of earlier sections, F-S decomposition of contingent claim is obtained in Section $6$. In Section $7$ we present a sensitivity analysis of the solution to the PDE. We end this paper by calculating the quadratic residual risk in Section $8$. Proof of some lemmata are given in Appendix.
	
	\section{Model description}\label{model}
	\noi This section consists of two subsections. In the first subsection we study a class of CSM processes. In the subsequent subsection we consider a market whose prices are governed by the CSM. Finally we show that, the market is arbitrage free by constructing an equivalent martingale measure.
	\subsection{Componentwise semi-Markov process}  Let $\mathcal{X}=\{1,\ldots,k\}\subset \mathbb{R}$. For every $l=0,1, \ldots, n$, consider a $C^1$ function $\la^l: \mathcal{X}\times  \mathcal{X}\times [0,\infty)\to (0,\infty)$ satisfying the following\\
	
	\textbf{Assumptions}(A1)
	\begin{itemize}
	\item[(i)] $\lambda^l_{ii}(\bar{y})=-\sum_{j\neq i}\lambda^l_{ij}(\bar{y})$,
	\item[(ii)]$\bar{y}\mapsto \la^l_{ij}(\bar{y})$ is continuously differentiable,
	\item[(iii)] if $\Lambda_i^l(\bar{y}):=\Int_0^{\bar{y}} \sum_{j\neq i} \lambda_{ij}^l(v)dv$, then $\underset{\bar{y}\rightarrow\infty}{\lim}\Lambda_i^l(\bar{y})=\infty.$
	
    \end{itemize}
    For each $l=0,\ldots,n$, let us consider a system of equations by replacing $\wp$ by $\wp^l$, $\la$ by $\la^l$ in \eqref{1} and \eqref{2}, where $\wp^l$ are independent Poisson random measures with intensity $dtdz$, defined on a complete probability space  $(\Omega,\mathfrak{F},\mathbb{P})$. That is,
	\begin{eqnarray}
	X^l_t &=& X^l_0 + \int_{(0,t]}\int_{\mathbb{R}} h^l(X^l_{s-}, Y^l_{s-},z)\wp^l(ds,dz)\text{\label{age1}}\\
	Y^l_t &=& t- \int_{(0,t]} \int_{\mathbb{R}} g^l(X^l_{s-}, Y^l_{s-},z) \wp^l(ds,dz),\text{\label{age2}}
	\end{eqnarray}
	where $h^l=h_{\la^l}$ and $g^l=g_{\la^l}$.
	It is shown in \cite{tanmay} using the results of \cite{IKW} that there exists an a.s
	unique strong solution to equations \eqref{age1}-\eqref{age2} and the process $Z^l_t:=(X^l_t,Y^l_t)$ is a time homogeneous strong Markov process. We denote the $\mathcal{X}^{n+1}$ valued process as $X_t$ whose $l$th component is $X^l_t$. Similarly we denote $Y_t=(Y_t^0,\ldots,Y_t^n)$. Thus $X_t$ is a CSM process.

	 Consider for each $l=0,\ldots,n$; $F^l:[0,\infty)\rightarrow [0,1]$ a differentiable function and defined as $F^l(\bar{y}|i):=1-e^{-\Lambda^l_i(\bar{y})}$, where $\Lambda^l_i$ is as in (A1)(iii). Let $f^l(\bar{y}|i):=\frac{d}{d\bar{y}}F^l(\bar{y}|i)$ and for each $j\neq i, p^l_{ij}(\bar{y}):=\frac{\la^l_{ij}(\bar{y})}{|\la^l_{ii}(\bar{y})|}$ with $p^l_{ii}(\bar{y})=0$ for all $i$ and $\bar{y}$. Set $\hat{p}^l_{ij}:=\Int_0^\infty p^l_{ij}(\bar{y}) dF^l(\bar{y}|i)$. In addition to (A1)(i)-(iii), we make the following   assumption
	 \begin{itemize}
		\item[(A1)] (iv) the matrix $(\hat{p}^l_{ij})_{k\times k}$ is irreducible.\\
	\end{itemize}  From the definition of $F^l$ and assumptions (A1)(ii)-(iii), we observe $0<F^l(\bar{y}|i)<1\quad \forall \bar{y}>0$ and $F^l(\bar{y}|i)\uparrow 1$ as $\bar{y}\to \infty$. It can easily be verified that $\la^l_{ij}(\bar{y})=p^l_{ij}(\bar{y})\frac{f^l(\bar{y}|i)}{1-F^l(\bar{y}|i)}$ hold for $i\neq j$. Let $T^l_n$ denote the time of $n^{\text{th}}$ transition of $X^l_t$ and $n^l(t)$ denote the total number of transitions upto time $t$ of $X^l_t$ i.e. $n^l(t):=\max\{n:T^l_n\leq t\}$. Hence $T^l_{n^l(t)}\leq t \leq T^l_{n^l(t)+1}$ and from \eqref{age2} $Y^l_t=t-T^l_{n^l(t)}$. It is shown in \cite{GhS} that $F^l(.|i)$ is the conditional c.d.f of the holding time of $X^l$ and $p^l_{ij}(\bar{y})$ is the
	conditional transition probability matrix. Let $\tau^l(t)$ be the duration after which $X^l_t$ would have a transition. Note that $\tau^l(t)$ is independent of
	every component of $X$ other than $l$th one. Let $F_{\tau^l}(\cdot|i,\bar{y})$ be the conditional c.d.f of $\tau^l(t)$ given $X^l_t=i$ and $Y^l_t=\bar{y}$. We note that this c.d.f does not depend on $t$ since $(X_t,Y_t)$ is time homogeneous. Therefore, $\tau^l(t)+Y^l_t$ is the duration of $X^l_t$ at present state between two transitions. Let $\ell(t)$ be the component of $X_t$ where the subsequent jump happens. Let $F_{\tau^l|l}(\cdot|x,y)$ be the conditional c.d.f of $\tau^l(t)$ given $X_t=x, Y_t=y$ and $\ell(t)=l$ and  $f_{\tau^l|l}(\cdot|x,y)$ be the conditional p.d.f of $\tau^l(t)$ given $X_t=x, Y_t=y$ and $\ell(t)=l$. From now we denote $P(\cdot|X_t=x,Y_t=y)$ by $P_{t,x,y}(\cdot)$ and the corresponding conditional expectation as $E_{t,x,y}(\cdot)$. We wish to compute $P_{t,x,y}(\ell(t)=l)$ i.e. the conditional probability of observing next jump to occur at the $l$th component given, $X_t=x$ and $Y_t=y$.
	 We compute this probability and some other conditional distribution and density functions in the following lemma.
	
	\begin{lem}\label{ptxyltl}
	Let $P_{t,x,y}(\ell(t)=l), F_{\tau^l|l}(v|x,y), f_{\tau^l|l}(v|x,y)$ be as above. Then the following hold
		\begin{enumerate}[(i)]
			\item $P_{t,x,y}(\ell(t)=l)= \Int_0^\infty\displaystyle\prod_{m\neq l}\frac{1-F^m(s+y^m|x^m)}{1-F^m(y^m|x^m)}\frac{f^l(s+y^l|x^l)}{1-F^l(y^l|x^l)}ds,$
			\item
		
		$F_{\tau^l|l}(v|x,y)= \frac{\Int_0^v\displaystyle\prod_{m\neq l}(1-F^m(s+y^m|x^m))f^l(s+y^l|x^l)ds}{\Int_0^\infty\displaystyle\prod_{m\neq l}(1-F^m(s+y^m|x^m))f^l(s+y^l|x^l)ds}$ , and\\
			$f_{\tau^l|l}(v|x,y)=\frac{\prod_{m\neq l}
				(1-F^m(v+y^m|x^m))f^l(v+y^l|x^l)}{\Int_0^\infty\prod_{m\neq l}(1-F^m(s+y^m|x^m))f^l(s+y^l|x^l)ds}$.
			Furthermore, $f_{\tau^l|l}(v|x,y)$ is differentiable with respect to $v$ and  we denote the derivative by $f'_{\tau^l\mid l}(v\mid x,y)$ and
			\item $f_{\tau^l|l}(0|x,y)P_{t,x,y}(\ell(t)=l)=\frac{f^l(y^l|x^l)}{1-F^l(y^l|x^l)}=f_{\tau^l}(0|x^l,y^l)$.
		 \end{enumerate}
	\end{lem}
The proof can be found in the appendix.

\subsection{Asset price dynamics}	
 We assume that $r:\mathcal{X}^{n+1}\to [0,\infty),$ $\mu^l: [0,T]\times \mathcal{X}^{n+1}\to \mathbb{R}$, and $\si^l:[0,T]\times \mathcal{X}^{n+1}\to \mathbb{R}^n$ are continuous functions for each $l=1,\ldots,n$. We consider a frictionless market consisting of one locally risk free asset and $n$ risky assets which may be referred to as stocks. Let $S_t^0$ be the price of money market account, with floating interest rate $r(X_t)$ at time $t$. Therefore its value at time $t$ is given by
	\begin{equation}
	dS_t^0 = r(X_t) S_{t}^0 dt, \quad S_0^0 = 1.
	\end{equation}
	The prices of the $l^{\text{th}}$ stock governed by $X_t$ is given by the following stochastic differential equation
	\begin{align} \label{sde}
	dS_t^l & = S_{t}^l \left[ \mu^l(t,X_{t} )dt +\sum_{j=1}^{n} \sigma^l_j(t,X_{t}) ~dW_t^j \right]   \\
	S_0^l &= s_l, \quad s_l \geq 0,  \nonumber
	\end{align}	
	where $\{W_t^j\}_{t\geq 0}$ are $n$ independent standard Wiener processes defined on $(\Omega,\mathfrak{F},P)$ independent of $\{\wp^l\}_{l=0}^n$.
 Here $\mu^l$ and $\sigma^l=(\sigma^l_1,\ldots,\sigma^l_n)$ represent the growth rate and volatility coefficient of $l^{\text{th}}$ asset respectively. We define the volatility matrix $\si(t,x):=(\si^l_{l'}(t,x))_{ll'}$ with $\si^l(t,x)$ its $l^{\text{th}}$ row vector and we denote $ (S_t^1,\dots,S_t^n) $ by $ S_t $ .  Let $\{\mathfrak{F}_t\}_{t\geq 0}$ be the completion of filtration generated by $S_t,X_t$ satisfying the usual hypothesis. Let $a(t,x) := \sigma(t,x) \sigma(t,x)^{*}=\left(\sum_{i=1}^{n}\si^l_i(t,x)\si^{l'}_i(t,x)\right)_{ll'}$ denote the diffusion matrix, where $*$ denotes the transpose operation. Then $a(t,x)$ is continuous on $[0,T]$. \\
 	\textbf{Assumption}(A1)(v) We assume that $\si(t,x)$ is invertible for each $(t,x) \in [0,T]\times\mathcal{X}^{n+1}$. 	
 	
 	\noi The  SDE \eqref{sde} has a unique strong solution with positive continuous paths and is given by
	\begin{align}\label{sdeSol}
	S_{t}^l & = s_l\exp\left[ \Int_{0}^{t  } \!\left(\mu^l(u, X_{u}) - \dfrac{1}{2} a^{ll}(u,X_{u})\right)du + \sum\nolimits_{j=1}^{n} \Int_{0}^{t  }\! \sigma^l_j(u, X_{u})\,dW_u^j  \right],~ \text{for}~ l\geq 1.
	\end{align}
	
	Then from \eqref{sdeSol},
	    		\begin{eqnarray*}
	    		\ln\frac{S^l_{t+v}}{S_t^l}=\int_t^{t+v} (\mu^l(u,X_u)-\frac{1}{2}a^{ll}(u,X_u))\,du + \int_t^{t+v} \sum_{j=1}^{n}\si^l_{j}(t,X_t)dW^j_t.
	    		\end{eqnarray*}
	    		
	    		\noi We define $Z:=(Z^1,\ldots,Z^n)$, where for each $l=1,\ldots,n$, $Z^l:=\ln\frac{S^l_{t+v}}{S_t^l}$. Clearly the conditional distribution of $Z$ given $S_t=s,X_t=x,Y_t=y,\ell(t)=m,\tau^m(t)=v$ is  conditional normal with mean $\bar{z}:=(\bar{z}^1,\ldots,\bar{z}^n)$, where
	    		\begin{eqnarray}\label{normalmean}
	    		\bar {z}^l:=\int_t^{t+v} (\mu^l(u,x)-\frac{1}{2} a^{ll}(u,x))\,du,
	    		\end{eqnarray}
	    		 and covariance matrix $\Sigma$ with
	    	$\Sigma^{ll'}:=cov\left(Z^l,Z^{l'}\right)$. i.e
	    	  \begin{eqnarray}\label{covsigma}
	    		\no \Sigma^{ll'} &=& E\left[\int_t^{t+v} \sigma^l(u,X_u)\,dW_u\times \int_t^{t+v} \sigma^{l'}(u,X_u)\,dW_u\Big|S_t=s,X_t=x,Y_t=y,\ell(t)=m,\tau^m(t)=v\right]\\
	    		&= & \int_t^{t+v}a^{ll'}(u,x)du.
	    		\end{eqnarray}
	    		In \eqref{normalmean} and \eqref{covsigma}, we have used the fact that the process $X$ remains constant on $[t,t+v)$ provided $\ell(t)=m,\tau^m(t)=v$ hold for some $ m $. We summarize the above derivation in the following Lemma where, we use a function $\theta:(0,\infty)^n\times (0,\infty)\times (0,\infty)^n\times \mathcal{X}^{n+1}\times (0,\infty)\to \mathbb{R}$ given by
	    		\begin{equation}\label{expalpha}
	    		\theta(\vs;t,s,x,v):=\frac{1}{\sqrt{(2\pi)^n|\Sigma|}\vs_1\vs_2\ldots\vs_n}
	    		\exp\left(-\frac{1}{2}\sum_{ll'}\Sigma^{-1}_{ll'}(z^l-\bar{z}^l)(z^{l'}-\bar{z}^{l'})\right),
	    		\end{equation}
	    		where $|\Sigma|$ is the determinant of $\Sigma$, $z^l=\ln(\frac{\vs^l}{s^l})$ and $s\in (0,\infty)^n,t\geq 0,x\in \mathcal{X}^{n+1},v>0$ and $\Sigma^{-1}_{ll'}$ is the $ll'$th element of $\Sigma^{-1}$ for $l=1,\ldots,n$.

   \begin{lem}\label{dists} If $S_t$ satisfies \eqref{sde}, then for any $v>0, t\geq 0$,

   \begin{itemize}
   \item[(i)]  $P\left(\left(\frac{S^l_{t+v}}{S^l_t}\leq \vs_l\right)_{l=1,\ldots, n}\bigg|S_t=s,X_t=x,Y_t=y,\ell(t)=m,\tau^m(t)=v\right)=\int_{\prod_{l=1}^{n}\left(0,\vs_l\right)}\theta(r;t,s,x,v)dr,$
   \item[(ii)] the conditional expectation is given by
   \begin{eqnarray*}
    E\left[\frac{S^l_{t+v}}{S^l_t} \bigg|S_t=s,X_t=x,Y_t=y,\ell(t)=m,\tau^m(t)=v\right]=e^{\int_{t}^{t+v}\mu^l(u,x)du},
    \end{eqnarray*}	
    \item[(iii)] the conditional covariance is given by $$cov\left(\frac{S^l_{t+v}}{S^l_t},\frac{S^{l'}_{t+v}}{S^{l'}_t}\bigg|S_t=s,X_t=x,Y_t=y,\ell(t)=m,\tau^m(t)=v\right)=e^{\int_{t}^{t+v}\left(\mu^l(u,x)+\mu^{l'}(u,x)\right)du}\left(e^{\int_{t}^{t+v}a^{ll'}(u,x)du}-1\right).$$
   \end{itemize}
  	    		
  	\end{lem}
  		\begin{lem}\label{Lem1}
  			Let $\{ S^l_t\}_{t\ge 0}$ be as in \eqref{sde} and $\{\mathcal{F}^X_t\}_{t\geq 0}$ be the filtration generated by $X$.
  			\begin{itemize}
  			\item[(i)] Then for each $l=1,\ldots,n$, and $t\geq 0,$
  			  			\begin{equation*}
  			  			E\left[  S^l_t \bigg|\mathcal{F}^X_t\right]\leq s_l e^{\int_0^t \mu^l(u,X_u)du}.
  			  			\end{equation*}
  			\item[(ii)] For each $l$, $E\left({S^l_t}^2\bigg| \mathcal{F}^X_t \right)<\infty$ for all $t$.
  			\end{itemize}
  		\end{lem}
  		\proof (i) Let $T^l_i$ be as in section 2.1,
  		\begin{eqnarray*}
  		E\left[ \frac{ S^l_t}{S^l_0} \bigg|\mathcal{F}^X_t\right] &=& E\left[\prod_{i=1}^{\infty} \frac{ S^l_{T^l_i\wedge t}}{ S^l_{T^l_{i-1}\wedge t}} \bigg|\mathcal{F}^X_t\right] \\
  		&=& E\left[\lim\limits_{N\to \infty}\prod_{i=1}^{N} \frac{ S^l_{T^l_i\wedge t}}{ S^l_{T^l_{i-1}\wedge t}} \bigg|\mathcal{F}^X_t\right] \\
  		&\leq&\varliminf_{N\to \infty} E\left[\prod_{i=1}^{N} \frac{ S^l_{T^l_i\wedge t}}{ S^l_{T^l_{i-1}\wedge t}} \bigg|\mathcal{F}^X_t\right],
  		\end{eqnarray*}
  		by Fatou's lemma. Now since for each $i=1,\ldots, n$, $\frac{ S^l_{T^l_i\wedge t}}{S^l_{T^l_{i-1}\wedge t}}$ are conditionally independent to each other given time $t$, and using Lemma \eqref{dists}(ii) the above limit can be rewritten as
  		$ \varliminf\limits_{N\to \infty} \prod_{i=1}^{N}e^{\int_{T^l_{i-1}\wedge t}^{T^l_i\wedge t}\mu^l(u,X_u)du},$ which is same as $e^{\int_0^t r(X_u)du}$.
  		
  		(ii) In a similar line of proof (i), using  Lemma \eqref{dists}(iii), the proof follows.
  	 \qed

  We denote the joint process $(\hat{S}^1_t,\ldots,\hat{S}^n_t)$ by $\hat{S}_t$, where $\hat{S}_t^l$ is given by $(S_t^0)^{-1}S_t^l$ and represents the discounted $l^{\text{th}}$ stock price. For each $l$,
	\begin{align}
	d\hat{S}_t^l & = \hat{S}_{t}^l \left[ \sum_{j=1}^{n} \sigma^l_j(t,X_{t}) ~dW_t^j +\left(\mu^l(t,X_{t} )-r(X_t)\right)dt\right],\label{8}
	\end{align}
	with $\hat{S}_0^l=s_l$.
	
	 To show that the market is arbitrage free under admissible strategy, we seek existence of an equivalent martingale measure (\cite{KK}, Th. 7.1). 	
	Consider $\gamma_l(t,x):=\sum_{j=1}^{n}{(\sigma^{-1}(t,x))}^l_j\left(\mu^j(t,x)-r(x) \right)$ for each $l=1,\ldots,n$. Under (A1)(v)and the continuity assumption on parameters, the Novikov's condition (\cite{IKW}, Th. 5.3) holds, i.e., for every $t\in [0,T],$
	$$E\bigg[\exp\left(\frac{1}{2}\sum_{l=1}^{n}\int_0^t\gamma^2_l(u,X_u)\,du\right)\bigg]<\infty.$$ Hence
	\begin{align*}
	\mathcal{Z}_t:=\exp\left(-\sum_{l=1}^{n} \int_0^t\gamma_l(u,X_u)\,dW^j_u-\frac{1}{2}\sum_{l=1}^{n}\int_0^t\gamma^2_l(u,X_u)\,du\right),
	\end{align*}
	is a square integrable martingale and $E\mathcal{Z}_T=1$. Consider an equivalent measure $P^*$ defined by $dP^*=\mathcal{Z}_TdP$. It is easy to check that $P^*$ is a probability measure. Hence by Girsanov's Theorem (\cite{KK}, Th. 5.5) $\bar W_t$ is a Wiener process under the probability measure $P^*$, where $\bar W^l_t=W^l_t+\int_{0}^{t}\gamma_l(u,X_u)du$. Thus \eqref{sde} becomes
	\begin{align*}
	dS_t^l & = S_{t}^l \left[ r(X_{t} )dt +\sum_{j=1}^{n} \sigma^l_{j}(t,X_{t}) ~d\bar W_t^j \right].
	\end{align*}
	Therefore under $P^*$, the discounted stock price $\hat{S}_t^l$ is a martingale and hence $P^*$ is an equivalent martingale measure. This proves that the market has no arbitrage under admissible strategies. The class of admissible strategy is presented in the next section.


	\section{Pricing Approach}\label{section for pricing approach}
	If $\xi^l_t$ denotes the number of units invested in the $l^{\text{th}}$ stock at time $t$ and $\vp_t$ denotes the number of units of the risk free asset, then $\pi=\{\pi_t=(\xi_t,\vp_t)\}_{t\in [0,T]}$ is called a portfolio strategy. For $t\in [0,T]$, $V_t(\pi) := \sum_{l=1}^n\xi^l_t S^l_t + \varepsilon_tS^0_t$ is said to be value process of the portfolio and the discounted value process is given by
	\begin{equation*}
	\hat{V}_t(\pi)=\xi_t \hat{S}_t +\vp_t.
	\end{equation*}
	\begin{defi}
A portfolio strategy $\pi=\{\pi_t=(\xi_t,\varepsilon_t), 0\leq t\leq T\}$ is called \emph{admissible} if it satisfies the following conditions
	\begin{itemize}
		\item[(i)] $\xi_t=(\xi^1_t,\ldots,\xi^n_t)$ is an $n$-dimensional predictable process and for each $l=1,\ldots,n$,
 $$\sum_{ll'}\int_0^T {\xi_t^l}{S^l_t} a^{ll'}(t,X_t)S^{l'}_t\xi^{l'}_t dt<\infty.$$
		\item[(ii)] $\vp$ is adapted, and $E(\varepsilon^2_t)<\infty~\forall~ t\in [0,T]$.
		\item[(iii)] $P(\hat{V}_t(\pi)\geq -a,~\forall~ t)=1$ for some positive $a$.
	\end{itemize}
	\end{defi}

 An \emph{admissible} strategy $\pi$ is called hedging strategy for an $\mathfrak{F}_T$ measurable claim $H$ if $V_T(\pi)=H$. For example, the claim associated to a European call option on $S^1$ is $H=(S^1_T-K)^+$, where $K$ is the strike price and $T$ is the maturity time. To price and hedge an option, an	 investor prefers an admissible hedging strategy which requires minimal amount of additional cash flow. In \cite{FS} the notion of ``optimal strtegy" is developed based on this idea. There the initial capital is referred as locally risk-minimizing price of the option. It is shown in \cite{FS} that if the market is arbitrage free, the existence of an optimal strategy for hedging a claim $H$, is equivalent to the existence of F\"{o}llmer Schweizer decomposition of discounted claim $\hat{H}:= {S^0_T}^{-1} H$ in the form
	\begin{equation}\label{eq1}
	\hat{H}=H_0+\sum_{l=1}^{n}\int^{T}_{0}{\xi_t^{\hat{H}}(l)}d{\hat{S}^l_t}+L^{\hat{H}}_T,
	\end{equation}
	where $H_0\in L^2(\Omega,\mathfrak{F}_0,P), L^{\hat{H}}=\{L^{\hat{H}}_t\}_{0\leq t\leq T}$ is a square integrable martingale starting with zero and orthogonal to the martingale part of $S_t$, and ${\xi_t^{\hat{H}}(l)}=(\xi_t^{\hat{H}}(1),\ldots,\xi_t^{\hat{H}}(n))$ satisfies A2 (i). Further ${\xi^{\hat{H}}(l)}$, appeared in the decomposition, constitutes the optimal strategy. Indeed the optimal strategy $\pi=(\xi_t,\varepsilon_t)$ is given by
	\begin{align}\label{11a}
	\nonumber \xi_t :=& \xi^{\hat{H}}_t,\\
	\hat{V}_t :=& H_0+\sum_{l=1}^{n}\int^{t}_{0}{\xi^l_u}\,d{\hat{S}^l_u}+L^{\hat{H}}_t,\\
	\nonumber \varepsilon_t :=& \hat{V}_t-\sum_{l=1}^{n}\xi^l_t {\hat{S}^l_t},
	\end{align}
	 and $S^0_t \hat{V}_t$ represents the \emph{locally risk minimizing price} at time $t$ of the claim $H$. Thus the F\"{o}llmer Schweizer decomposition is the key thing to settle the pricing and hedging problems in a given market. We refer to \cite{S2} for more details. In this paper we are interested to price a special class of contingent claims, of the form $H=K(S_T)$, where we make the following assumptions on $K:\mathbb{R}^n_+\to \mathbb{R}_+$.
	
	 \noi  \textbf{Assumptions} (A2):
	 	  \begin{itemize}
	 	  \item[(i)] $K(s)$ is Lipschitz continuous function.
	 	  \item[(ii)] There exists $c_1\in \mathbb{R}^n,~ \text{and}~ c_2>0$ such that $|K(s)-c_1^*s|<c_2$ for all $s\in\mathbb{R}^n_+$.
	 	  \end{itemize}  This class includes claims of all types of basket options consisting finitely many vanilla options. As an example, a typical basket call option has a claim $(\sum_{l=1}^n c_lS^l_t-\bar{K})^+$, where $\bar{K}$ is the strike price.

	  Our primary goal in this paper is to obtain expressions for locally risk-minimizing price process and the optimal strategy corresponding to a claim $K(S_T)$.
	  To this end we study the Cauchy problem \eqref{mainpde}-\eqref{pdeboundary} and obtain expressions of price and hedging using solution of \eqref{mainpde}-\eqref{pdeboundary}. We state this result as theorems at the end of this section. But before that we introduce some notation and definition. We define a linear operator $$D_{t,y}\vf(t,s,x,y):=\displaystyle\lim_{\vp\rightarrow 0}\frac{1}{\vp}\{\vf(t+\vp,s,x,y+\vp\mathbf{1})-
\vf(t,s,x,y)\},$$
where dom($D_{t,y}$), the domain of $D_{t,y}$, contains all measurable functions $\vf$ on $\mathcal{D}$ such that above limit exists for every $(t,s,x,y)\in \mathcal{D}$. We rewrite \eqref{mainpde} using the above notation
	\begin{align}\label{p1}
 \no & D_{t,y} \varphi(t, s, x, y) +r(x)\sum_{l=1}^n s^l \frac{\pa \varphi}{\pa s^l}(t, s, x, y)+\frac{1}{2}\sum_{l=1}^{n}\sum_{l'=1}^{n}a^{ll'}(t,x)s^ls^{l'}\frac{\partial^2\vf}{\partial s^{l}\pa s^{l'}}(t,s,x,y)  \\
 	&+ \sum_{l=0}^n \sum_{j\neq x^l}\lambda_{x^lj}^l(y^l)\left[ \vf(t,s,R^l_j x,R^l_0 y)-\vf(t,s,x,y) \right] = r( x) ~ \varphi(t, s, x, y).
	\end{align}
 Now we define the meaning of classical solution of the PDE.
\begin{Definition}\label{defcl}
		We say, $\varphi:\mathcal{D}\rightarrow \mathbb{R}$ is a classical solution of \eqref{p1}-\eqref{pdeboundary} if $\varphi\in$ \emph{dom}$(D_{t,y})$, twice differentiable with respect to $s$ and for all $(t,s,x,y)\in \mathcal{D}$, \eqref{p1}-\eqref{pdeboundary} are satisfied.
	\end{Definition}
	
	\begin{theo}\label{existence theorem}
		Under assumptions \emph{(A1)(i)-(v)}, the initial value problem \eqref{p1}-\eqref{pdeboundary} has a unique classical solution in the class of functions with at most linear growth.
	\end{theo}
	
We establish this at the end of section \ref{section for pde }. We present the locally risk-minimizing strategy in terms of the solution to the PDE \eqref{p1}-\eqref{pdeboundary}. The proof of the following Theorem is deferred to Section \ref{section for pricing}.
	\begin{theo}\label{theo5} Let $\vf$ be the unique classical solution of \eqref{p1}-\eqref{pdeboundary} in the class of functions with at most linear growth and $(\xi,\varepsilon)$  be given by
				\begin{equation}\label{VI3.20}
				\xi^l_t :=\frac{\partial\varphi}{\partial s^l}(t,S_t,X_{t-},Y_{t-})~\forall~l=1,\ldots,n, \text{ and } \varepsilon_{t} := e^{-\int_{0}^{t}r(X_{u})du} \left(\varphi(t,S_t,X_{t},Y_{t})-\sum_{l=1}^{n}\xi^l_{t}S^l_{t}\right).
				\end{equation}
				Then
	\begin{enumerate}
	\item
		 $(\xi,\varepsilon)$ is the optimal admissible strategy,
   \item $\varphi(t,S_t,X_t,Y_t)$ is the locally risk minimizing price of the claim $K(S_T)$ at time $t$.
	
	\end{enumerate}
	\end{theo}
	
	In order to study the well-posedness of solution of the PDE \eqref{p1}-\eqref{pdeboundary}, we study a Volterra integral equation of second kind. We prepare ourself by showing the existence and uniqueness of solution of the integral equation in the next section.
	

	\section{Volterra Integral Equation}\label{existenceofsol}
	For each $x$, consider the following Cauchy problem which is known as B-S-M PDE
		\begin{align}\label{black scholes}
		\frac{\partial\rho_{x}(t,s)}{\partial t} + r(x)\sum_{l=1}^{n} s^l\frac{\partial\rho_{x}(t,s)}{\partial s^l}+\frac{1}{2}\sum_{l=1}^{n} \sum_{l'=1}^{n}a^{ll'}(t,x)s^ls^{l'} \frac{\partial^2\rho_{x}(t,s)}{\partial s^l \pa s^{l'}} = r(x) \rho_{x}(t,s)
		\end{align}
		for $(t,s)\in (0,T)\times (0,\infty)^n$ and $\rho_{i}(T,s)=K(s)$, for all $s\geq 0$. This has a classical solution with at most
		linear growth (see \cite{KK}), provided $K$ is of at most linear growth.
	We would like to mention that $\rho_x$ is infinitely many times differentiable with respect to $s$.

		For $ \zeta\in\RR^n $, let $ \|\zeta\|_1 $ denote the norm
$\sum_{l=1}^{n}|\zeta^l| $. Let
\begin{eqnarray}\label{classB}
 \mathcal{B}:=\{\vf:\bar{\mathcal{D}}\to \mathbb{R}, \mathrm{measurable}\mid \|\vf\|_L:=\sup_{\bar{\mathcal{D}}}\frac{|\vf(t,s,x,y)|}{1+\|s\|_1}<\infty \} .
\end{eqnarray}
Let $C^2_s(\mathcal{D}):=C^{0,2,0}(\mathcal{D})$ be the set of all measurable functions on $\mathcal{D}$, which are also twice differentiable with respect to $s$.

     Let $\Sigma$ be an $n\times n$ matrix, whose elements are as in \eqref{covsigma}. We further use the notation $\Sigma$, $|\Sigma|$ and $\Sigma^{-1}$ as in \eqref{expalpha}. By replacing $\mu^l(u,x)$ by $r(x)$ in \eqref{normalmean}, we define a function $\alpha:(0,\infty)^n\times (0,\infty)\times (0,\infty)^n\times \mathcal{X}^{n+1}\times (0,\infty)\to \mathbb{R}$ by
     	\begin{equation}\label{newalpha}
         		\alpha(\vs;t,s,x,v)=\frac{1}{\sqrt{(2\pi)^n|\Sigma|}\vs_1\vs_2\ldots\vs_n}
         		\exp\left(-\frac{1}{2}\sum_{ll'}\Sigma^{-1}_{ll'}(z^l-\bar{z}^l)(z^{l'}-\bar{z}^{l'})\right),
         		\end{equation}
         		where $z^l=\ln(\frac{\vs^l}{s^l})$, and $\bar{z}^l:=\int_t^{t+v} (r(x)-\frac{1}{2}a^{ll}(u,x))\,du $ for $s\in (0,\infty)^n,t\geq 0,x\in \mathcal{X}^{n+1},v>0$ and $\Sigma^{-1}_{ll'}$ is the $ll'$th element of $\Sigma^{-1}$ for $l=1,\ldots,n$. It is clear from \eqref{newalpha} that $\alpha(\vs;t,s,x,v)$ is a log-normal density with respect to $\vs$ variable for a fixed $(t,s,x,v)$. The mean of the corresponding log-normal distribution is $e^{\int_{t}^{t+v}r(x)du}=e^{r(x)v}$.

    Consider the following integral equation
     	\begin{align}\label{main integral eqn}
     		\varphi(t,s,x,y)=&\no \sum_{l=0}^{n} P_{t,x,y}(\ell(t)=l) \left( \rho_x(t,s)\left(1-F_{\tau^l\mid l}(T-t\mid x,y)\right)+ \int_0^{T-t}e^{-r(x) v} f_{\tau^l\mid l}(v\mid x,y)\times\right. \\
     		&\sum_{j\neq x^l}p^l_{x^l j}(y^l+v)\left.\int_{\mathbb{R}_+^n} \vf\left(t+v,\vs,R^l_jx,R^l_0(y+v\mathbf{1})
     		\right)\alpha(\vs;t,s,x,v)\,d\vs\,dv \right).
     		\end{align}
     	\begin{lem}\label{integralsol}
     				The integral equation \eqref{main integral eqn} has a unique solution in $\mathcal{B}$ \em(as in \eqref{classB}).
     				
     	\end{lem}	
     	\begin{proof}
     	 We first note that a solution of \eqref{main integral eqn} is a fixed point of the operator $A$ and vice versa, where
     		\begin{align*}
     		\no A \varphi(t,s,x,y)&:=\no \sum_{l=0}^{n} P_{t,x,y}(\ell(t)=l) \left( \rho_x(t,s)\left(1-F_{\tau^l\mid l}(T-t\mid x,y)\right) + \int_0^{T-t}e^{-r(x) v} f_{\tau^l\mid l}(v\mid x,y)\times\right. \\
     		&\sum_{j\neq x^l}p^l_{x^l j}(y^l+v)\left.\int_{\mathbb{R}_+^n} \vf\left(t+v,\vs,R^l_jx,R^l_0(y+v\mathbf{1})
     		\right)\alpha(\vs;t,s,x,v)\,d\vs\,dv \right).
     		\end{align*}
     		
     		\noi  It is simple to verify that for each $\vf\in \mathcal{B}$, $A\vf: \bar{\mathcal{D}}\to \mathbb{R}$ is measurable. To prove that $A$ is a contraction in $\mathcal{B}$, we need to show that for $\vf_1,\vf_2\in \mathcal{B}$, $\|A\varphi_1-A\varphi_2\|_L \leq J\|\varphi_1-\varphi_2\|_L$ where $J<1$. In order to show existence and uniqueness in the prescribed class, it is sufficient to show that $A$ is a contraction in $\mathcal{B}$. Then the Banach fixed point Theorem ensures existence and uniqueness of the fixed point in $\mathcal{B}$. To show that for $\vf_1,\vf_2\in \mathcal{B}$, $\|A\varphi_1-A\varphi_2\|_L \leq J\|\varphi_1-\varphi_2\|_L$ where $J<1$, we compute
     		\begin{align*}
     		\|A\varphi_1-A\varphi_2\|_L=&\sup_{\mathcal{D}}\bigg|\frac{ A\varphi_1-A\varphi_2}{1+\|s\|_1}\bigg|\\
     		=&\sup_{\mathcal{D}}\bigg| \sum_{l=0}^{n}P_{t,x,y}(\ell(t)=l) \int_0^{T-t}e^{-r(x) v} f_{\tau^l\mid l}(v\mid x,y)\sum_{j^l\neq x^l}p^l_{x^l j^l}(y^l+v)\times\bigg.\\
     		&\bigg.\int_{\RR^n_+}(\vf_1-\vf_2)(t+v,\vs,R^l_jx,R^l_0(y+v\mathbf{1}))\frac{\alpha(\vs;t,s,x,v)}{1+\|s\|_1}\,d\vs\,dv  \bigg|\\
     		\leq & \sup_{\mathcal{D}}\bigg| \sum_{l=0}^{n}P_{t,x,y}(\ell(t)=l) \int_0^{T-t}e^{-r(x) v} f_{\tau^l\mid l}(v\mid x,y)\sum_{j^l\neq x^l}p^l_{x^l j^l}(y^l+v)\times\bigg.\\
     		&\bigg.\int_{\RR^n_+}(1+\|\vs\|_1)\sup_{(t',\vs',x',y')\in\mathcal{D}}\left[\frac{(\vf_1-\vf_2)(t',\vs',x',y')}{1+\|\vs'\|_1}\right]\frac{\alpha(\vs;t,s,x,v)}{1+\|s\|_1}\,d\vs\,dv  \bigg|\\
     		=& \sup_{\mathcal{D}}\bigg| \sum_{l=0}^{n}P_{t,x,y}(\ell(t)=l)\int_0^{T-t}e^{-r(x) v} f_{\tau^l\mid l}(v\mid x,y)\|\vf_1-\vf_2\|_L\frac{\bar{\alpha}(t,s,x,v)}{1+\|s\|_1}\,dv  \bigg|,
     		\end{align*}
     		where $\bar{\alpha}(t,s,x,v):= \int_{\RR_+^n} (1+\|\vs\|_1) \alpha(\vs;t,s,x,v)\,d\vs$. Replacing $\mu^l(u,x)$ by $r(x)$ in Lemma \ref{Lem1}(i), we get $$\bar{\alpha}(t,s,x,v) = 1+ \|s\|_1 e^{r(x) v}.$$
     		
     		\noi Thus, $\|A\varphi_1-A\varphi_2\|_L \, \leq\, J\|\varphi_1-\varphi_2\|_L$, where
     		
     		\begin{align*}
     		J=& \sup_{\mathcal{D}}\bigg| \sum_{l=0}^{n}P_{t,x,y}(\ell(t)=l) \int_0^{T-t}e^{-r(x) v} f_{\tau^l\mid l}(v\mid x,y)\frac{1+ \|s\|_1 e^{r(x) v}}{1+\|s\|_1}\,dv  \bigg|\\
     		\leq & \sup_{\mathcal{D}}\bigg| \sum_{l=0}^{n}P_{t,x,y}(\ell(t)=l) \int_0^{T-t} f_{\tau^l\mid l}(v\mid x,y)\,dv  \bigg|\\
     		=& \sup_{\mathcal{D}}\bigg| \sum_{l=0}^{n}P_{t,x,y}(\ell(t)=l)  F_{\tau^l\mid l}(v\mid x,y)  \bigg|\\
     		<& \sup_{\mathcal{D}}\bigg| \sum_{l=0}^{n}P_{t,x,y}(\ell(t)=l)  \bigg|=1,
     		\end{align*}
     		using $r(x)\geq 0$ and the fact that $ F^l(y|i)<1 $ for all $l,x,y $ and $ i $. Thus $A$ is a contraction in $\mathcal{B}$. This completes the proof.\qed
     	\end{proof}	
     	
     	\begin{remark}
     	By a direct substitution $t=T$ in the \eqref{main integral eqn}, we obtain $\vf(T,s,x,y)=K(s)$. It is interesting to note that we do not have to impose any other boundary conditions for existence and uniqueness of solution of \eqref{main integral eqn}. We can directly obtain other boundary values by substituting the boundary in the integral equation.
     	\end{remark}
     	Next we state some differentiability results of the coefficients \eqref{main integral eqn} in the following Lemma. The Proof of this lemma is given in the appendix.
     		\begin{lem}\label{regularity} Let $F_{\tau^l|l}(T-t|x,y)$ and $P_{t,x,y}(\ell(t)=l)$ be as in Lemma \ref{ptxyltl} and $\alpha(\vs;t,s,x,v)$  as in \eqref{newalpha}. Then
     		\begin{enumerate}
     		\item[(i)] $F_{\tau^l|l}(T-t|x,y)$ and $P_{t,x,y}(\ell(t)=l)$ are  in \emph {dom}($D_{t,y}$), and
     	\begin{align*}
     	 D_{t,y}P_{t,x,y}(\ell(t)=l) = &\sum_{r=0}^{n} f_{\tau^r}(0|x^r,y^r)P_{t,x,y}(\ell(t)=l)-f_{\tau^l}(0|x^l,y^l)\\
     	 D_{t,y}F_{\tau^l|l}(T-t|x,y)= &f_{\tau^l|l}(0|x,y) (F_{\tau^l|l}(T-t|x,y)-1).
     	\end{align*}
     	\item[(ii)] $f_{\tau^l|l}(t|x,y)$ is in \emph{dom}($D_{t,y}$).
     	\item[(iii)] $\alpha(\vs;t,s,x,v)$ is $C^1$ in $t,v$, and infinite time differentiable in $s$.
     	\end{enumerate}
     		\end{lem}
     	\begin{lem}\label{lm2}
     			Let $ \vf \in \mathcal{B} $ be the solution of the integral equation \eqref{main integral eqn}. Then (i) $\vf\in$ \emph{dom}$(D_{t,y})\cap C^2_s(\mathcal{D})$, and (ii) $ \vf(t,s,x,y) $ is non-negative.
     		\end{lem}
     		\proof (i) Using the smoothness of $\rho_x$ for each $x$, the first term on the right hand side of \eqref{main integral eqn} is in dom$(D_{t,y})\cap C^2_s(\mathcal{D})$. Thus it is enough to check the desired smoothness of
     			\begin{equation*}
     			\beta_l(t,s,x,y)=\int_0^{T-t}e^{-r(x) v} f_{\tau^l\mid l}(v\mid x,y)
     			\sum_{j\neq x^l}p^l_{x^l j}(y^l+v)\int_{\mathbb{R}_+^n} \vf\left(t+v,\vs,R^l_jx,R^l_0(y+v\mathbf{1})\right)\alpha(\vs;t,s,x,v)\,d\vs\,dv .
     			\end{equation*}
   First we check the applicability of $D_{t,y}$. It is easy to see that $D_{t,y}\beta_l(t,s,x,y)$ is the limit of the following expression
     			\begin{align*}
     			\lefteqn{\frac{1}{\vp}\Big[\int_0^{T-t-\vp}e^{-r(x) v} f_{\tau^l\mid l}(v\mid x,y+\vp\mathbf{1})
     				\sum_{j\neq x^l}p^l_{x^l j}(y^l+v+\vp)\int_{\mathbb{R}_+^n} \vf\left(t+v+\vp,\vs,R^l_jx,R^l_0(y+(v+\vp)\mathbf{1})\right)\times}\\
     			&\alpha(\vs;t+\vp,s,x,v)\,d\vs\,dv
     			-\int_0^{T-t}e^{-r(x) v} f_{\tau^l\mid l}(v\mid x,y)
     			\sum_{j\neq x^l}p^l_{x^l j}(y^l+v)\int_{\mathbb{R}_+^n} \vf\left(t+v,\vs,R^l_jx,R^l_0(y+v\mathbf{1})\right)\times\\
     			&\alpha(\vs;t,s,x,v)\,d\vs\,dv \Big].
     			\end{align*}
     			After a suitable substitution, the above expression becomes
     			\begin{align}\label{betabar}
     			\no \lefteqn{\int_{\vp}^{T-t}e^{-r(x) v}\sum_{j\neq x^l}p^l_{x^l j}(y^l+v)\int_{\mathbb{R}_+^n}\vf\left(t+v,\vs,R^l_jx,R^l_0(y+v\mathbf{1})\right)\bar{\beta}_{\vp}(v,\vs;t,s,x,y)d\vs\,dv\Big.}\\
     			&-\frac1\vp\int_{0}^{\vp}  e^{-r(x)v}f_{\tau^l|l}(v|x,y)\sum_{j\neq x^l}p^l_{x^l j}(y^l+v)\int_{\mathbb{R}_+^n}\vf\left(t+v,\vs,R^l_jx,R^l_0(y+v\mathbf{1})\right)\alpha(\vs;t,s,x,v)\,d\vs\,dv,
     			\end{align}
     			where $$\bar{\beta}_{\vp}(v,\vs;t,s,x,y):=\frac{1}{\vp}\big(e^{r(x)\vp}f_{\tau^l\mid l}(v-\vp\mid x,y+\vp\mathbf{1})\alpha(\vs;t+\vp,s,x,v-\vp)
     			-f_{\tau^l\mid l}(v\mid x,y)\alpha(\vs;t,s,x,v)\big).$$
     			Now the above defined $\bar{\beta}_{\vp}$ can be rewritten as
     			\begin{align}\label{mvt}
     		\no\lefteqn{\frac{1}{\vp}\Big[\left(e^{r(x)\vp}-1+1\right)\big(f_{\tau^l\mid l}(v-\vp\mid x,y+\vp\mathbf{1})-f_{\tau^l\mid l}(v\mid x,y+\vp\mathbf{1})+f_{\tau^l\mid l}(v\mid x,y+\vp\mathbf{1})-f_{\tau^l\mid l}(v\mid x,y)}\\
     			+&\no f_{\tau^l\mid l}(v\mid x,y)\big)\times\big(\alpha(\vs;t+\vp,s,x,v-\vp)-\alpha(\vs;t,s,x,v-\vp)
     			+\alpha(\vs;t,s,x,v-\vp)-\alpha(\vs;t,s,x,v)
     			+\alpha(\vs;t,s,x,v)\big)\\
     			-& f_{\tau^l\mid l}(v\mid x,y)
     			\alpha(\vs;t,s,x,v)\Big].
     			\end{align}
Due to the continuous differentiability results in Lemma \ref{ptxyltl}(ii) and Lemma \ref{regularity}(ii),(iii), we can use the mean value Theorem to rewrite \eqref{mvt} as
     			\begin{align*}
     			&\Big[\left(\vp r(x)e^{r(x)\vp_0}+1 \right)\left(-f'_{\tau^l\mid l}(v-\vp_1\mid x,y+\vp\mathbf{1})+\sum_{i=1}^{n}\frac{\pa}{\pa y_i} f_{\tau^l\mid l}(v\mid x,y+\vp_2\mathbf{1})+\frac{1}{\vp}f_{\tau^l\mid l}(v\mid x,y)\right)\times\\
     			&\Big(\vp\alpha_t(\vs;t+\vp_3,s,x,v-\vp)
     			-\vp\alpha_v(\vs;t,s,x,v-\vp_4)
     			+\alpha(\vs;t,s,x,v)\Big)-\frac{1}{\vp}f_{\tau^l\mid l}(v\mid x,y)
     			\alpha(\vs;t,s,x,v) \Big],
     			\end{align*}
     			for some $\vp_0,\vp_1,\vp_2,\vp_3,\vp_3<\vp$. After some rearrangement of terms in the above expression, we get
     			\begin{align*}
     			 \bar{\beta}_{\vp}(v,\vs;t,s,x,y)=&\alpha(\vs;t,s,x,v)\left(r(x)e^{r(x)\vp_0}f_{\tau^l|l}(v|x,y)-f'_{\tau^l\mid l}(v-\vp_1\mid x,y+\vp\mathbf{1})+\sum_{i=1}^{n}\frac{\pa}{\pa y_i} f_{\tau^l\mid l}(v\mid x,y+\vp_2\mathbf{1})\right)\\
     			&+f_{\tau^l\mid l}(v\mid x,y)\Big( \alpha_t(\vs;t+\vp_3,s,x,v-\vp)
     			-\alpha_v(\vs;t,s,x,v-\vp_4) \Big)+\vp \mathcal{G}_{\vp}(v,\vs;t,s,x,y),
     			\end{align*}
     	 where \begin{align*}
                  \mathcal{G}_{\vp}(v,\vs;t,s,x,y):=& r(x)e^{r(x)\vp_0}\left(-f'_{\tau^l\mid l}(v-\vp_1\mid x,y+\vp\mathbf{1})+\sum_{i=1}^{n}\frac{\pa}{\pa y_i} f_{\tau^l\mid l}(v\mid x,y+\vp_2\mathbf{1})\right)\times\\
                  &\left(\vp\alpha_t(\vs;t+\vp_3,s,x,v-\vp)-\vp\alpha_v(\vs;t,s,x,v-\vp_4)+\alpha(\vs;t,s,x,v)\right)\\
                 & + \left(r(x)e^{r(x)\vp_0}f_{\tau^l\mid l}(v\mid x,y)-f'_{\tau^l\mid l}(v-\vp_1\mid x,y+\vp\mathbf{1})+\sum_{i=1}^{n}\frac{\pa}{\pa y_i} f_{\tau^l\mid l}(v\mid x,y+\vp_2\mathbf{1})\right)\times\\
                 &\left(\alpha_t(\vs;t+\vp_3,s,x,v-\vp)-\alpha_v(\vs;t,s,x,v-\vp_4)\right).
     	     \end{align*}
    \noi We also recall from \eqref{alphat} and \eqref{alphav} that $$\alpha_t(\vs;t,s,x,v)=\alpha(\vs;t,s,x,v) O(\log^2|\vs|)~	\text{and}~ \alpha_v(\vs;t,s,x,v)=\alpha(\vs;t,s,x,v) )O(\log^2|\vs|),$$
    where $|\vs|:=\max_{i}|\vs_i|$.
The expression in \eqref{betabar} has two additive terms. For showing convergence of the first term, we intend to use above expressions for applying dominated and Vitali's convergence theorems in various cases. For that, as $\vf \in \mathcal{B}$, it would be sufficient if we have the following three results,
    \begin{itemize}
   \item[(a)] $v\mapsto \int_{\mathbb{R}^n_+}\left(c_1^*\vs+c_2\right) \log^2|\vs|\alpha(\vs;t,s,x,v)d\vs$ is bounded and left continuous,
   \item[(b)] $t\mapsto \int_{\mathbb{R}^n_+}\left(c_1^*\vs+c_2\right) \log^2(|\vs|)\alpha(\vs;t,s,x,v)d\vs$ is continuous uniformly with respect to $v$,
   \item [(c)] $\|\vs\|^2\alpha(\vs;t+\vp_1,s,x,v+\vp_2)$  is uniform integrable and tight w.r.t. $\vs$ for $\vp_1, \vp_2 \ll1$.
    \end{itemize}
    To prove the result (a), we introduce a function $B(v):=\int_{\mathbb{R}^n_+}\left(c_1^*\vs+c_2\right) \log^2(|\vs|)\alpha(\vs;t,s,x,v)d\vs.$ Now for $\vp>0$ using the mean value theorem, there exist a $0<\vp'<\vp$ such that
    \begin{align*}
   \frac{1}{\vp}\left( B(v)-B(v-\vp)\right)=&\int_{\mathbb{R}^n_+}\left(c_1^*\vs+c_2\right) \log^2(|\vs|)\alpha_v(\vs;t,s,x,v-\vp')d\vs\\
    \leq &  \int_{\mathbb{R}^n_+}\left(c_3\|\vs\|^2_2+c_4\right) \alpha(\vs;t,s,x,v-\vp')d\vs,
    \end{align*}
for some positive constants $c_3,c_4$. Now Lemma \ref{dists}(iii) suggests that the right hand side is bounded in $v$ on $[\vp,T]$. This implies that $B$ is left continuous. Using the similar reasoning the boundedness of $B$ also follows from Lemma \ref{dists}(iii).  Similarly one can prove the result (b). In order to prove (c), we first recall that a family of normal random variables with bounded mean and variance is uniformly integrable and tight. Therefore (c) follows as here a product of a polynomial and a lognormal density function appears.

Now we address the convergence of the second term of \eqref{betabar}. Clearly the result (a) implies boundedness of $v\mapsto \int_{\mathbb{R}_+^n}\vf\left(t+v,\vs,R^l_jx,R^l_0(y+v\mathbf{1})\right)\alpha(\vs;t,s,x,v)\,d\vs$, which assures the desired convergence. Thus $ \beta_l\in \text{dom}(D_{t,y}) $ and hence $ \vf\in \text{dom}(D_{t,y})$.

\noi Now we discuss the smoothness with respect to $s$. First we observe that $\alpha_{s^{l'}}(\vs;t,s,x,v)=\frac{1}{s^{l'}}O(\log(|\vs|))\alpha(\vs;t,s,x,v)$. Since $\vf \in \mathcal{B}$, using uniform integrability and tightness of $\frac{1}{s^{l'}+\vp}\|\vs\|^2\alpha(\vs;t,s+\vp,x,v)$ and uniform boundedness of $v\mapsto \int_{\mathbb{R}^n_+}\frac{1}{s^{l'}+\vp}\|\vs\|^2\alpha(\vs;t,s+\vp,x,v)d\vs$ for $\vp\ll1$, we conclude the differentiability of $\beta_l(t,x,y)$ with respect to $s^l$. Similarly we can establish existence of partial derivatives of any higher order successively. Thus one can obtain twice continuous differentiability of $\beta_l$.
     			
    \noi (ii) We have already shown that $ A:\mathcal B\rightarrow\mathcal B $ is a contraction. From the form of equation \eqref{black scholes}, and non-negativity of $K$, it is clear that \eqref{black scholes} admits a non-negative solution. Since all the coefficients in equation \eqref{main integral eqn} are non-negative, it follows that $ A\vf\geq 0 $ for $ \vf\geq 0 $. Furthermore, we have shown that $ A $ has a fixed point in $ \mathcal{B} $. It can be easily argued that this fixed point is, in fact, non-negative. Hence, $ \vf  $ is non-negative.
     			 \qed
    \section{The Partial Differential Equation}\label{section for pde }
    In this section we establish Theorem \ref{existence theorem}, i.e uniqueness and existence of \eqref{p1}-\eqref{pdeboundary}. Before addressing that it is important to clarify few issues regarding boundary conditions. At $s=0$ facet the partial derivative with respect to $s$ disappear. Since the nature of the domain is triangular, it can be shown by using the method of characteristic that the initial condition would lead to a solution to \eqref{p1}-\eqref{pdeboundary}. It can also be shown that the PDE would have no solution if we impose a boundary condition which is not obtain from the initial condition. We refer (\cite{tanmay},pp.32) for more details.
     Let $\tilde{W}$ be a standard $n$-dimensional Brownian motion on a probability space $(\tilde\Omega,\tilde{\mathcal{F}}, \tilde P)$. For each $l=1,2,\dots,n$, let $\tilde{S}^l_t$ satisfies
    \begin{equation}\label{emmsde}
    d\tilde{S}^l_t = \tilde{S}^l_t\left[r(X_t)dt + \sum_{j=1}^{n}\si^l_{j}(t,X_t)d\tilde{W}^j_t\right],\tab \tilde{S}_0>0,
    \end{equation}
    where $X_t$ is the age-dependent process given by equations \eqref{age1} and \eqref{age2} on  $(\tilde\Omega,\tilde{\mathcal{F}}, \tilde P)$ and $\si^l$ is the $l$th row of $\si$. We denote $\tilde{S}_t:=(\tilde{S}^1_t,\ldots,\tilde{S}^n_t)$.

    	\begin{prop}\label{theo4} (i) The Cauchy problem \eqref{p1}-\eqref{pdeboundary} has a generalized solution, $\vf$. (ii) Under assumption \emph{(A1)(i)-(iv}), $\vf$ solves the integral equation \eqref{main integral eqn}.(iii) $\vf \in \mathcal{B}$.
    	\end{prop}	
    	
    	\proof (i) Let $\tilde{S}_t$ be the strong solution of the SDE \eqref{emmsde}. Let $\tilde{ \mathfrak{F}}_t$ be the filtration generated by $ \tilde{S}_t $ and $ X_t $, that satisfies the usual hypothesis. Since $(t,X_t,Y_t)$ is Markov, then the process $(t,\tilde{S}_t,X_t,Y_t)$ is Markov process. Let $\mathcal{A}$ be the infinitesimal generator of $(t,\tilde{S}_t,X_t,Y_t)$ , where 	
    		\begin{align}\label{defgen}
    	\no	\mathcal{A} \varphi(t,s,x,y)=&D_{t,y}\vf(t,s,x,y)+r(x)\sum_{l=1}^n s^l \frac{\pa \varphi}{\pa s^l}(t, s, x, y)+\frac{1}{2}\sum_{l=1}^{n}\sum_{l'=1}^{n}a^{ll'}(t,x^l)s^ls^{l'}\frac{\partial^2\vf}{\partial s^{l}\pa s^{l'}}(t,s,x,y) \\
    		&+ \sum_{l=0}^n \sum_{j\neq x^l}\lambda_{x^lj}^l(y^l)\left[ \vf(t,s,R^l_jx,R^0_ly)-\vf(t,s,x,y) \right],
    		\end{align}
    		for every function $\varphi$ which is compactly supported $C^2$ in $s$ and $C^1$ in $y$.
    		Let \begin{equation}\label{n_t}
    		N_t:=  E[e^{-\int_t^T r(X_u)du} K(\tilde{S}_T)\mid \tilde{S}_t=s,X_t=x,Y_t=y].
    		\end{equation}
    The above expectation is finite due to (A2)(ii) and Lemma \ref{Lem1}. Thus  \eqref{n_t} suggests that $N_t$ is a $\tilde{\mathfrak{F}}_t$ martingale. Since $K(s)$ has at-most linear growth, and $\tilde{S}_t$ has finite expectation, \eqref{n_t} suggests that $ E|N_t| < \infty$ for each $t$. Hence using the Markov semigroup of $(t,\tilde{S}_t,X_t,Y_t)$ the PDE has a generalized solution $\vf:\mathcal{D}\to \mathbb{R}$ measurable given by
     \begin{equation}\label{gensol}
        		\vf(t,s,x,y):=  E[e^{-\int_t^T r(X_u)du} K(\tilde{S}_T)\mid \tilde{S}_t=s,X_t=x,Y_t=y].
        		\end{equation}
    	
   \noi (ii) 		By conditioning \eqref{gensol} on transition times, we get
    		\begin{eqnarray*}
    		\lefteqn{\varphi(t,\tilde{S}_t,X_t,Y_t)}\\
    		&=& E\left[ E\left[ e^{-\int_t^T r(X_u)\,du} K(\tilde{S}_T)\mid \tilde{S}_t,X_t,Y_t,l(t)=l \right]\mid \tS_t,X_t,Y_t \right]\\
    		&=& \sum_{l=0}^{n} P_{t,x,y}(\ell(t)=l) \ E\left[ e^{-\int_t^T r(X_u)\,du} K(\tilde{S}_T)\mid \tilde{S}_t,X_t,Y_t,l(t)=l \right]\\
    		&=& \sum_{l=0}^{n} P_{t,x,y}(\ell(t)=l) \ E\left[ E\left[ e^{-\int_t^T r(X_u)\,du} K(\tilde{S}_T)\mid \tilde{S}_t,X_t,Y_t,l(t)=l,\tau^l(t) \right]\mid \tilde{S}_t,X_t,Y_t,l(t)=l \right].
    		\end{eqnarray*}
    		Now,
    		\begin{eqnarray*}
    		\lefteqn{E\left[ E\left[ e^{-\int_t^T r(X_u)\,du} K(\tilde{S}_T)\mid \tilde{S}_t,X_t,Y_t,l(t)=l,\tau^l(t) \right]\mid \tilde{S}_t,X_t,Y_t,l(t)=l \right]}\\
    		&=& P[\tau^l(t)>T-t]\rho_x(t,\tS_t)+\int_0^{T-t} E\left[ e^{-\int_t^T r(X_u)\,du} K(\tilde{S}_T)\mid \tilde{S}_t,X_t,Y_t,l(t)=l,\tau^l(t)=v \right]f_{\tau^l\mid l}(v\mid X_t,Y_t)\,dv.
    		\end{eqnarray*}
    		We note that
    		\begin{eqnarray*}
    		\lefteqn{E\left[ e^{-\int_t^T r(X_u)\,du} K(\tilde{S}_T)\mid \tilde{S}_t,X_t,Y_t,l(t)=l,\tau^l(t)=v \right]}\\
    		&=& e^{-r(X_t)v}\sum_{j^l\neq X^l}p^l_{X^l j^l}(Y^l+v) \int_{\mathbb{R}_+^n} E\left[ e^{-\int_{t+v}^T r(X_u)\,du} K(\tilde{S}_T)\mid \tS_{t+v}=\vs,X_{t+v}=R^l_jx,\right.\\
    		&&\left. Y_{t+v}=R^l_0y,l(t)=l,\tau^l(t)=v \right]\alpha(\vs;t,s,x,v)\,d\vs.
    		\end{eqnarray*}
    		Therefore
    		\begin{eqnarray}
    		\lefteqn{\varphi(t,\tilde{S}_t,X_t,Y_t)}\no\\
    		&=&\no \sum_{l=0}^{n}P_{t,x,y}\left(\ell(t)=l\right)
    	 \left( \rho_x(t,\tS_t)\left(1-F_{\tau^l\mid l}(T-t\mid X_t,Y_t)\right)  + \int_0^{T-t}e^{-r(X_t)v} f_{\tau^l\mid l}(v\mid X_t,Y_t)\times\right.\\
    		&& \no \left.\sum_{j\neq x^l}p^l_{x^l j}(y^l+v) \int_{\mathbb{R}_+^n} \vf\left(t+v,R^l_{\vs^l}s,R^l_jx,R^l_0y\right)\alpha(\vs;t,s,x,v)\,d\vs\,dv \right).
    		\end{eqnarray}
    		Using (A1)(iv), and since $ \lambda^l_{ij}(y)>0 $ for $ i\neq j $, we can replace $( \tilde{S}_t, X_t,Y_t)$ by the generic variable $(s,x,y)$ in the above relation. As a conclusion, $\varphi$ is a solution of \eqref{main integral eqn}.

    		\noi (iii) To show $\vf$ is of at-most linear growth, it is sufficient to show for all $(t,s,x,y)\in \mathcal{D}$ $|\vf(t,s,x,y)-c^*_1s|\leq c_2,$ where $c_1,c_2$ is as in (A2)(ii).
    		 We note that, if $\tilde{S}_t$ is the solution of \eqref{emmsde}, $e^{-\int_{0}^{t}r(X_u)du}\tilde S_t$ is a $\tilde{\mathfrak{F}}_t$ martingale. Therefore by using the Markov property of $\tilde{S}_t,X_t,Y_t$, and  the fact $e^{-\int_{0}^{t}r(X_u)du}$ is $\tilde{\mathfrak{F}}_t$-measurable, we obtain
    		\begin{align*}
            E\left[e^{-\int_{t}^{T}r(X_u)du}\tilde S_T\big|\tilde{S}_t,X_t,Y_t\right]=&	E\left[e^{-\int_{t}^{T}r(X_u)du}\tilde S_T\big|\tilde{\mathfrak{F}}_t\right]\\
            =& e^{-\int_{0}^{t}r(X_u)du}E\left[e^{-\int_{0}^{T}r(X_u)du}\tilde S_T\big|\tilde{\mathfrak{F}}_t\right]\\
    		=& \tilde{S}_t.
          \end{align*}
           Using this equality, \eqref{gensol} and (A2)(ii), we have
              		\begin{align*}
              		&|\vf(t,s,x,y)-c^*_1s|\\
              		&=\bigg|E\left[e^{-\int_t^T r(X_u)du} K(\tilde{S}_T)\mid\tilde{S}_t=s,X_t=x,Y_t=y\right]- c^*_1 E\left[e^{-\int_t^T r(X_u)du} \tilde{S}_T\mid \tilde{S}_t=s,X_t=x,Y_t=y\right]\bigg|\\
              		&\leq  E\left[[e^{-\int_t^T r(X_u)du}|K(\tilde{S}_T)-c^*_1 \tilde{S}_T |\mid \tilde{S}_t=s,X_t=x,Y_t=y\right]\\
              		&\leq  c_2.
              		\end{align*}
              		This completes the proof.\qed

  \noi \textbf{Proof of Theorem \ref{existence theorem}:} Proposition \ref{theo4} implies that the PDE \eqref{p1}-\eqref{pdeboundary} has a generalized solution which is in $\mathcal{B}$, and also solves the integral equation. Lemma \ref{integralsol} suggests that the integral equation has only one solution in $\mathcal{B}$. Finally Lemma \ref{lm2} asserts that this unique solution of the integral equation is in dom$(D_{t,y})\cap C^2_s$. Therefore using the above results, we conclude that \eqref{p1}-\eqref{pdeboundary} has a generalized solution which is in the domain of the operators in \eqref{p1}. Hence the generalized solution \eqref{gensol} solves \eqref{p1}-\eqref{pdeboundary} classically. To prove the  uniqueness, first assume that $\vf_1$ and $\vf_2$ are two classical solutions of \eqref{p1}-\eqref{pdeboundary} in the prescribed class of functions. Then using Proposition \ref{theo4}, it follows that both also solve \eqref{main integral eqn}. By Lemma \ref{integralsol}, there is only one such solution in the prescribed class. Hence $\vf_1=\vf_2$.\qed	

    \begin{lem}\label{deribnd}
         Let $\vf(t,s,x,y)$ be the classical solution of the Cauchy problem \eqref{p1}-\eqref{pdeboundary}.Under assumption \emph{(A2)(i)}, $\frac{\pa \vf}{\pa s^m}(t,s,x,y)$ is bounded.
    \end{lem}
    \proof Since $\vf(t,s,x,y)$ is the classical solution of \eqref{p1}-\eqref{pdeboundary} it is in dom$(D_{t,y})\cap C^2_s$. In fact $\vf$ has greater regularity than $C^2_s$ which is evident in the proof of Lemma \ref{lm2}. Indeed due to Lemma \ref{regularity}(iii) and the $C^{\infty}$ smoothness of $\rho, \vf$ is $C^{\infty}$ in $s$.  Let $\psi^m(t,s,x,y):=\frac{\pa \vf}{\pa s^m}(t,s,x,y)$, for $m=1,\ldots,n$. Now differentiating equation \eqref{p1} with respect to $s^m$ and using the fact that $a(t,x)$ is symmetric, we obtain
    \begin{align}\label{derivf}
   \no D_{t,y}\psi^m(t,s,x,y)+&\sum_{l=1}^{n}s^l\left(r(x)+a^{ml}(t,x)\right)\frac{\pa \psi^m}{\pa s^l} (t,s,x,y)+
    \frac{1}{2}\sum_{l=1}^{n}\sum_{l'=1}^{n}a^{ll'}(t,x)s^ls^{l'}\frac{\pa^2 \psi^m}{\pa s^l \pa s^{l'}}(t,s,x,y)\\
    +&\sum_{l=0}^n \sum_{j\neq x^l}\lambda_{x^lj}^l(y^l)\left[ \psi^m(t,s,R^l_j x,R^l_0 y)-\psi^m(t,s,x,y) \right]=0.
    \end{align}	
    It is easy to check that
     \begin{align*}
       \no \hat{\mathcal{A}}\psi^m(t,s,x,y)&=D_{t,y}\psi^m(t,s,x,y)+\sum_{l=1}^{n}s^l\left(r(x)+a^{ml}(t,x)\right)\frac{\pa \psi^m}{\pa s^l} (t,s,x,y)\\
       &+\frac{1}{2}\sum_{l=1}^{n}\sum_{l'=1}^{n}a^{ll'}(t,x)s^ls^{l'}\frac{\pa^2 \psi^m}{\pa s^l \pa s^{l'}}(t,s,x,y)+\sum_{l=0}^n \sum_{j\neq x^l}\lambda_{x^lj}^l(y^l)\left[ \psi^m(t,s,R^l_j x,R^l_0 y)-\psi^m(t,s,x,y) \right],
        \end{align*}
is the infinitesimal generator of the Markov process $(t,\bar{S}_t,X_t,Y_t)$, where $\bar{S}_t=(\bar{S}_t^1,\ldots,\bar{S}_t^n)$ and $\bar{S}_t^l$ satisfies the following SDE
\begin{eqnarray}\label{newsde}
d\bar{S}_t^l=\bar{S}_t^l\left[\left(r(X_t)I+\text{Diag}(a^l(t,X_t)\right)dt+\si(t,X_t)dW_t\right],
\end{eqnarray}
where Diag$(a^l(t,X_t)$ is the diagonal matrix containing the $l^{\text{th}}$ row of $a(t,x)$ and $(X_t,Y_t)$ is as in \eqref{age1}-\eqref{age2}. Therefore the solution of the PDE \eqref{derivf} has the stochastic representation of the following form
\begin{eqnarray}\label{repofsol}
\psi^m(t,s,x,y)=E\left[K'(\bar{S}^m_T)\bigg|\bar{S}^m_t=s,X_t=x,Y_t=y\right],
\end{eqnarray}
where $K':\mathbb{R}^n_+\to \mathbb{R}$ is defined almost everywhere by $K(s)=\int_{0}^{s^m}K'(R^m_rs)dr$, for each $s\in\mathbb{R}^n_+$
Since $K$ is of at-most linear growth and it is Lipschitz continuous, $K'$ is in $L^{\infty}$. Hence \eqref{repofsol} suggests $\psi^m(t,s,x,y)$ is bounded.
\qed

	\section{Pricing and Hedging}\label{section for pricing}
 \textbf{Proof of Theorem \ref{theo5}:}
   Using Lemma \ref{deribnd} we can show that $\pi=(\xi,\vp)$ as given in \eqref{VI3.20} is an admissible portfolio strategy. Indeed $\xi^l_t$ is left continuous and therefore predictable. Hence (A2)(i) and (ii) holds for this pair  $\pi=(\xi,\vp)$. Therefore the discounted value function for this pair of strategy using \eqref{VI3.20} is given by
	\begin{equation*}
	\hat{V}_t(\pi)=\sum_{l=1}^{n}\xi^l_{t}\hat{S}^l_{t}+\vp_t=e^{-\int_{0}^{t}r(X_{u})\,du} \varphi(t,S_{t},X_{t},Y_{t}),
	\end{equation*}
	where $\vf$ is the unique classical solution of \eqref{p1}-\eqref{pdeboundary}.
	Now we shall find a decomposition for $\hat{V}_t(\pi)$. Under the
	measure $P$, we apply It\={o}'s formula to $$e^{-\int_{0}^{t}r(X_{u})\,du} \varphi(t,S_{t},X_{t},Y_{t}).$$ Using \eqref{8}, \eqref{p1} and \eqref{1} and after a suitable rearrangement of terms,  for all $t<T$, we obtain,
	\begin{eqnarray}\label{VI3.24}
	\no e^{-\int_{0}^{t}r(X_{u})\,du} \varphi(t,S_{t},X_{t},Y_{t}) & = &
	\varphi(0,S_{0},X_{0},Y_{0})+\sum_{l=1}^{n}\int_{0}^{t}\frac{\partial
		\varphi}{\partial s^l} (u,S_{u},X_{u-},Y_{u-})d {\hat{S}_u^l}\no
	\no +\int_{0}^{t}e^{-\int_{0}^{u}r(X_{v})\,dv}\\
	& &\no \hspace{-5mm}\int_{\mathbb{R}} [\varphi(u,S_{u},X_{u-} +h(X_{u-},Y_{u-},z),Y_{u-}-g(X_{u-},Y_{u-},z))\\
	&&-\varphi(u,S_{u},X_{u-},Y_{u-})]{\hat{\wp}}(du,dz),
	\end{eqnarray}
	where $\hat{\wp}$ is the compensator of $\wp$, i.e.  $\hat{\wp}(dt,dz)=\wp(dt,dz)-dtdz$.
	Therefore from \eqref{VI3.24}, we have for each $t\leq T$
	\begin{equation}\label{newfs}
	\frac{1}{S^0_t} \varphi(t,S_t,X_{t-},Y_{t-})= H_0+\sum_{l=1}^{n}\int^{t}_{0}{\xi^l_u}d{\hat{S}_u^l}+L_t,
	\end{equation}
   where $H_0=\vf(0,S_{0},X_{0},Y_{0})$ and
	\begin{eqnarray}\label{fsl}
	\no L_t&:=& \int_{0}^{t}e^{-\int_{0}^{u}r(X_{v})dv}\int_{\mathbb{R}} [\varphi(u,S_{u},X_{u-} +h(X_{u-},Y_{u-},z),Y_{u-}-g(X_{u-},Y_{u-},z)) \\
	&&-\varphi(u,S_{u},X_{u-},Y_{u-})]{\hat{\wp}}(du,dz).
	\end{eqnarray}
	Clearly the above choice of $H_0$ is  $\mathfrak{F}_0$ measurable and $L_T$ is  $\mathfrak{F}_T$ measurable. We know that, the integral with respect to a compensated Poisson random measure is a local martingale. Hence $L_t$ is a local martingale. The proof of Proposition \ref{theo4}(iii) suggests that expectation of supremum of $L_t$ is finite. Hence it is a martingale. Again since $W_t$ and $\wp$ are independent, $L_t$  is orthogonal to $\int_0^t \sig^l(t,X_t)\hat{S}_t dW_t$. Thus, we obtain the following F-S decomposition by letting $t \uparrow T $ in \eqref{newfs},
	\begin{equation}\label{VI3.25}
	{S^0_T}^{-1}K(S_T) = \varphi(0,S_{0},X_{0},Y_{0}) + \sum_{l=1}^{n}\int^{T}_{0}{\xi^l_t}d{\hat{S}^l_t} + L_T.
	\end{equation}
	This completes the proof.\qed

	\begin{theo}
		Let $ \vf $ be the unique solution of \eqref{main integral eqn}.  Set
			\begin{align}
			\no\eta(t,s,x,y):=&\sum_{l=0}^{n}P_{t,x,y}(\ell(t)=l) \left( \frac{\partial \rho_x(t,s)}{\partial s^m}\left(1-F_{\tau^l\mid l}(T-t\mid x,y)\right) + \int_{0}^{T-t}e^{-r(x) v}f_{\tau^l|l}(v|x,y)\times\right.\\
			&\left.\sum_{j^l\neq x^l}p^l_{x^lj^l}(y^l+v)\int_{\RR_+^n}\vf(t+v,\vs,R^l_jx,R^0_ly)\frac{\pa\alpha(\vs;t,s,x,v)}{\pa s^m}\,d\vs\,dv \right),\label{20}
			\end{align}
			where $(t,s,x,y)\in \mathcal{D}$.
			Then $\eta(t,s,x,y)=\frac{\pa \vf}{\pa s^m}(t,s,x,y) $,
	\end{theo}
	\begin{proof}
		We need to show that $\psi$ (as in \eqref{20}) is equal to $\frac{\partial\varphi}{\partial s^m}$. Indeed, one obtains the RHS of \eqref{20} by differentiating the right side of \eqref{main integral eqn} with respect to $s^m$. Hence the proof. \qed
	\end{proof}
	
	\begin{remark}\label{rem1}
		
		\noi We have shown in Theorem \ref{theo5} that $ \frac{\pa \vf}{\pa s^m}(t,s,x,y)$ is a necessary quantity to be calculated in order to find the optimal hedging. Attempting to compute $ \frac{\pa \vf}{\pa s^m}(t,s,x,y)$ using numerical differentiation would increase the sensitivity of $ \frac{\pa \vf}{\pa s^m}(t,s,x,y)$ to small errors. Equation \eqref{20} gives a better, more robust approach for computing $ \frac{\pa \vf}{\pa s^m}(t,s,x,y)$, using numerical integration.
	\end{remark}


	\section{Sensitivity with respect to the instantaneous rate function}\label{section for senitivity}
	
	\noi In a recent paper Goswami et al. \cite{AGSN} gave an interesting idea to approximate the solution by approximating the transition rate. In the previous section we have seen that for a class of continuously differentiable transition rate function, there exists a unique classical solution of the PDE \eqref{p1}-\eqref{pdeboundary}. Let $\la:=(\la^0,\ldots,\la^n)$ be a vector where $\la^l$ is as in section \ref{model}.  We state and prove the important result below.
	
  \begin{theo}\label{3}
		 Let $\vf$ and $\tilde\vf$ be two solutions of \eqref{p1}-\eqref{pdeboundary} with parameter $\la$ and $\tilde \la$ respectively. Then
		 $\|\vf-\tilde\vf\|_{\text{sup}}\leq 2c_2 T\|\la-\tilde{\la}\|_{\text{sup}}$, where $c_2$ as in Assumption \emph{(A2)(ii)}.
 	\end{theo}

	\proof
 Let $\vf$ be the classical solution and  $\tilde\vf$ be its TBA. We consider
   \begin{equation}\label{defpsi}
   \psi(t,s,x,y) := \vf(t,s,x,y) - \tilde\vf(t,s,x,y).
   \end{equation}
	
	\noi Now, it is easy to see that $\psi$ satisfies the following initial  value problem,
	\begin{align}\label{difference}
	&\no D_{t,y}\psi(t,s,x,y) +r(x)\sum_{l=1}^n s^l \frac{\pa \psi}{\pa s^l}(t, s, x, y)+\frac{1}{2}\sum_{l=1}^{n}\sum_{l'=1}^{n}a^{ll'}(t,x)s^ls^{l'}\frac{\partial^2\psi}{\partial s^{l}\pa s^{l'}}(t,s,x,y)\\
	 &\no +\sum_{l=0}^{n}\sum_{j \neq x^l} \la^l_{x^l j}(y^l) \Big(\psi(t,s,R_j^l x,R_0^l y) -  \psi(t,s,x,y)\Big)\\
	 &= r(x) \psi(t,s,x,y) - \sum_{l=0}^{n}\sum_{j \neq x^l} \Big(\lambda^l_{x^l j}(y^l)-\tilde\lambda^l_{x^l j}(y^l)\Big)\Big( \tilde\vf(t,s,R_j^l x,R_0^l y) - \tilde\vf(t,s,x,y)\Big),
	\end{align}
	defined on
	$$
	\mathcal{D}:=\{(t,s,x,y) \in (0,T) \times \mathbb{R}^n_{+} \times \mathcal{X}^{n+1} \times (0,T)^{n+1} | y \in (0,t)^{n+1}\},
	$$
	with condition
	$$\psi(T,s,x,y) =0, ~~ s \in \mathbb{R}^n_{+}; ~~ 0 \le y^l \le T; ~~ x=1,2,\cdots,k.$$
	
	\noi We rewrite \eqref{difference} using \eqref{defpsi} as
	\begin{equation}\label{diffpde}
	\mathcal{A} \psi(t,s,x,y) = r(x) \psi(t,s,x,y) - f(t,s,x,y),
	\end{equation}
	where
	\begin{align*}
	f(t,s,x,y):=\sum_{l=0}^{n}\sum_{j \neq x^l} \Big(\lambda^l_{x^l j}(y^l)-\tilde\lambda^l_{x^l j}(y^l)\Big)\Big( \tilde\vf(t,s,R_j^l x,R_0^l y) - \tilde\vf(t,s,x,y)\Big).
	\end{align*}
	We recall that $\mathcal{A}$ is the infinitesimal generator of $(t,\tilde S_t,X_t,Y_t)$. Using the proof of Proposition \ref{theo4}(iii), one can show that for all $(t,s,x,y)\in \mathcal{D}$
		\begin{equation}\label{fbnd}
		|f(t,s,x,y)|\le 2c_2\sum_{l=0}^n\sum_{j\neq x^l}\|\lambda_{x^lj}^l(y)-\tilde\lambda_{x^lj}^l(y)\|_{\text{sup}}.
		\end{equation}
	 The stochastic representation of the solution of the PDE \eqref{diffpde} is given by,
	\begin{equation}
	\psi(t,s,x,y) = E[\int_t^T \exp\left(-\int_t^v r(X_u)du\right) f(v,\tilde{S}_v,X_v,Y_v)dv | \tilde S_t=s, X_t=x, Y_t=y].
	\end{equation}

	\noi Since $\tilde{\vf}$ is a solution of \eqref{p1}-\eqref{pdeboundary} for parameter $\tilde{\la}$, then the proof of Proposition \ref{theo4}(iii), $|\tilde\vf(t,s,R_j^l x,R_0^l y) - \tilde\vf(t,s,x,y)|<2c_2.$
	Now using \eqref{fbnd} and  $r>0$ for all $t\leq v\leq T$, we have
	\begin{eqnarray*}
\|\psi(t,s,x,y)\|_\text{sup} &=& \sup_{\bar{\mathcal{D}}}\Big| E[\int_t^T \exp\left(-\int_t^v r(X_u)du\right) f(v,\tilde{S}_v,X_v,Y_v)dv | \tilde S_t=s, X_t=x, Y_t=y]\Big|\\
	&\leq &2c_2(T-t)\sum_{l=0}^n\sum_{j\neq x^l}\|\lambda_{x^lj}^l(y)-\tilde\lambda_{x^lj}^l(y)\|_{\text{sup}}\\
	&<& 2c_2 T\sum_{l=0}^n\sum_{j\neq x^l}\|\lambda_{x^lj}^l(y)-\tilde\lambda_{x^lj}^l(y)\|_{\text{sup}}.
	\end{eqnarray*}
 Hence the proof is completed.\qed
	\begin{rem}
	It is interesting to note that a weaker variant of Theorem \ref{3} can also be proved if the assumption \emph{(A2)(ii)} is relaxed. Indeed if $K\in \mathcal{B}$ for such case $\|\vf-\tilde{\vf}\|_L\leq M\|\la-\tilde{\la}\|_{\text{sup}}$. This readily follows from the fact that $\tilde{\vf}$ is of at most  linear growth and $\tilde{S}_t$ has finite expectation.
	\end{rem}
	
	\section{Calculation of the Quadratic Residual Risk}\label{section for risk residual}
	In this section we find an expression of the quadratic residual of risk.
	Let $\pi:=(\xi_t,\vp_t)$, where $\xi_t=(\xi^1_t,\ldots,\xi^n_t)$ be the admissible strategy and $V_t$ be the value process as defined in Section \ref{section for pricing approach}. Further we assume that $ \{C_t\}_{t\geq0} $ be the accumulated additional cash flow process associated with the optimal hedging of the contingent claim $ H $, where
		\begin{eqnarray*}
		dC_t=dV_t-\sum_{l=1}^{n}\xi^l_t dS^l_t-\vp_tdS^0_t.
		\end{eqnarray*}
		One can show that
		\begin{eqnarray}\label{cash}
		\frac{1}{S^0_t}dC_t=d\hat{V}_t-\sum_{l=1}^{n}\xi^l_t d\hat{S}^l_t,
		\end{eqnarray}
		where $\hat{V}_t$ is the discounted value process as defined in Section \ref{section for pricing approach}. Now by \eqref{11a}, we have
		\begin{eqnarray}\label{deriF-S}
		d\hat{V}_t=\sum_{l=1}^{n}\xi^l_t d\hat{S}^l_t+dL^{\hat{H}}_t.
		\end{eqnarray}
		Now comparing \eqref{cash} and \eqref{deriF-S}, we have
		\begin{eqnarray*}
		\frac{1}{S^0_t}dC_t=dL^{\hat{H}}_t.
		\end{eqnarray*}
	The discounted value, at $ t=0 $, of the accumulated cash flow during $ [0,T] $ is
	\begin{equation*}
	\hat{C}_T-\hat{C}_0:=\int_0^T \frac{1}{S^0_t}\,dC_t=L^{\hat{H}}_t.
\end{equation*}
Using  above and \eqref{fsl}, we get
	\begin{equation*}
	L_T^{\hat{H}}=\int_0^T\frac1{S^0_t}\int_{\RR}\left(\vf(t,S_t,X_{t-}+h(X_{t-},Y_{t-},z),Y_{t-}-g(X_{t-},Y_{t-},z)-\vf(t,S_t,X_{t-},Y_{t-})\right)\,\hat{\wp}(dt,dz).
   \end{equation*}
	Thus
	\begin{equation}\label{36}
   dC_t=\int_{\RR}\left(\vf(t,S_t,X_{t-}+h(X_{t-},Y_{t-},z),Y_{t-}-g(X_{t-},Y_{t-},z)-\vf(t,S_t,X_{t-},Y_{t-})\right)\,\hat{\wp}(dt,dz).
	\end{equation}
	Integrating the above expression, we obtain the external cash flow associated with the optimal hedging.
		Hence,
			\begin{align}\label{cashflow}
		\no	C_T=& C_0+\int_0^T \int_{\RR}\left(\vf(t,S_t,X_{t-}+h(X_{t-},Y_{t-},z),Y_{t-}-g(X_{t-},Y_{t-},z)-\vf(t,S_t,X_{t-},Y_{t-})\right)\,\hat{\wp}(dt,dz)\\
			\no =& C_0 + \sum_{t\in[0,T]}\left( \vf(t,S_t,X_t,Y_t)-\vf(t,S_t,X_{t-},Y_{t-}) \right) - \int_0^T \sum_{l=0}^{n}\sum_{j\neq X^l_{t-}} \lambda^l_{X^l_{t-}j}(Y^l_{t-})\times\\
			&\left[ \vf(t,S_t,R^l_jX_{t-},R^l_0Y_{t-}) -\vf(t,S_t,X_{t-},Y_{t-}) \right]\,dt,
			\end{align}
	Given a strategy $ \pi $, we can define the \emph{quadratic residual risk} at $ t=0 $, denoted by $ \mathcal{R}_0(\pi) $, which is given by $ \mathcal{R}_0(\pi):=E[(\hat{C}_T-\hat{C}_0)^2|\mathcal{F}_0] $.
	
	\begin{lem}\label{theo 4.1}
	 The quadratic variation process of $ C_t $ is given by
		\begin{equation}
		[C]_t=\sum_{r\in[0,t]}\left(\vf(r,S_r,X_r,Y_r)-\vf(r,S_r,X_{r-},Y_{r-}) \right)^2,
		\end{equation}
		where $ \vf $ is the unique classical solution of \eqref{p1}-\eqref{pdeboundary} with at most linear growth.
	\end{lem}
	\begin{proof}
  It is clear that $ C_t $ as in \eqref{cashflow} is an \emph{rcll} process. Now, for $ r\in(0,T) $ and for sufficiently small $ \Delta $, we have
		\begin{align*}
		\lefteqn{(C_r-C_{r-\Delta})^2= \left( \vf(r,S_r,X_r,Y_r)-\vf(r,S_r,X_{r-\Delta},Y_{r-\Delta}) \right)^2 - 2\left(\vf(r,S_r,X_r,Y_r)-\vf(r,S_r,X_{r-\Delta},Y_{r-\Delta})  \right) \times}\\
		& \sum_{l=0}^{n}\sum_{j\neq X^l_{r-\Delta}}\lambda^l_{X^l_{r-\Delta}j}(Y^l_{r-\Delta})\left[ \vf(r,S_r,R^l_jX_{r-\Delta},R^l_0Y_{r-\Delta}) -\vf(r,S_r,X_{r-\Delta},Y_{r-\Delta}) \right]\Delta\\
		&+\left(\sum_{l=0}^{n}\sum_{j\neq X^l_{r-\Delta}} \lambda^l_{X^l_{r-\Delta}j}(Y^l_{r-\Delta}) \left[ \vf(r,S_r,R^l_jX_{r-\Delta},R^l_0Y_{r-\Delta} -\vf(r,S_r,X_{r-\Delta},Y_{r-\Delta}) \right]\right)^2 \Delta^2+ O\left({\Delta^2}\right).
		\end{align*}
		Since the quadratic variation of $ C_t $ is the limit of the sum $ \sum_{r\in[0,t]}(C_r-C_{r-\Delta})^2 $ over a partition with $ \Delta\rightarrow 0 $, we take the summation both sides. We note that the second term, the multiplier of $\Delta$ is bounded and is of $O(\Delta)$ except a set of whose measure is $O(\Delta)$, Thus the summation of second, third and fourth terms in the above expression can be ignored. Hence,
		\begin{equation}
		[C]_t=\sum_{r\in[0,t]}\left[\vf(r,S_r,X_r,Y_r)-\vf(r,S_r,X_{r-},Y_{r-}) \right]^2.
		\end{equation}\qed
	\end{proof}
	
	\noi An expression for $ \mathcal{R}_0(\pi) $ can be found using It\=o's isometry. Further using Lemma \ref{theo 4.1}, we get
	\begin{align}
	\no \mathcal{R}_0(\pi)=& E[(\hat{C}_T-\hat{C}_0)^2|\mathcal{F}_0]\\
	\no =& E\left[\left(\int^T_0 \frac1{S^0_t}\,dC_t \right)^2 \mid \mathcal{F}_0 \right]\\
	\no =& E\left[\int^T_0 \frac{1}{{S^0_t}^2}\,d[C]_t\mid \mathcal{F}_0 \right]\\
	\no=& E\left[ \sum_{t\in[0,T]}\frac{1}{{S^0_t}^2}\left(\vf(t,S_t,X_t,Y_t)-\vf(t,S_t,X_{t-},Y_{t-})\right)^2\mid \mathcal{F}_0 \right]\\
	\no=& E\left[ \frac{1}{{S^0_t}^2}\left(\hat{\vf}(t,S_t,X_t,Y_t)-\hat{\vf}(t,S_t,X_{t-},Y_{t-})\right)^2\mid \mathcal{F}_0 \right]\\
	=& E\left[\sum_{m=1}^{m(T)}\left(\hat{\vf}(T_m,S_{T_m},X_{T_m},Y_{T_m})-\hat{\vf}(T_m,S_{T_m},X_{T_{m-1}},T_m-T_{m-1})\right)^2  \mid \mathcal{F}_0 \right]\label{38}
	\end{align}

	\section*{Appendix}
	
	\textbf{Proof of Lemma \ref{ptxyltl}:}\\
	\noi (i) In order to compute $P_{t,x,y}(\ell(t)=l)$, we first derive  the conditional c.d.f $F_{\tau^l}(\cdot|i,\bar{y})$.
	\begin{eqnarray}\label{eqf}
	\no F_{\tau^l}(s|i,\bar{y})&=&P(0\leq\tau^l(t)\leq s|X^l_t=i,Y^l_t=\bar{y})\\
	\no &=& P(\tau^l(t)+Y^l_t\leq s+\bar{y}|X^l_t=i,Y^l_t=\bar{y})\\
	\no &=& P(Y^l_{T_{n^l(t)+1}-}\leq s+\bar{y}|Y^l_{T_{n^l(t)}-}\geq \bar{y},X^l_t=i,Y^l_t=\bar{y})\\
	& = & \frac{F^l(s+\bar{y}|i)-F^l(\bar{y}|i)}{1-F^l(\bar{y}|i)}\quad\quad l=0,1,\ldots,n .
	\end{eqnarray}
	Let $F_{\tau^l}(s|i,\bar{y}):=f_{\tau^l}(s|i,\bar{y})$. Therefore
	\begin{eqnarray}\label{eqft}
	f_{\tau^l}(\cdot|i,\bar{y})=\frac{f^l(\cdot+\bar{y}|i)}{1-F^l(\bar{y}|i)}.
	\end{eqnarray}
	Let $F_{\tau^{-l}}(\cdot|x,y)$ denotes conditional c.d.f of $\tau^{-l}(t)$ given $X_t=x$ and $Y_t=y$ as is given by $1-\displaystyle\prod_{m\neq l} \bigg(1-F_{\tau^m}(\cdot|x^m,y^m)\bigg)$, where $\tau^{-l}(t):=\displaystyle\min_{m\neq l}\tau^m(t)$.

	Therefore it follows from the definition of $\ell(t)$ that,  $P_{t,x,y}(\ell(t)=l)=P_{t,x,y}(\tau^l(t)<\tau^{-l}(t))$. We compute this probability  using conditioning on $\tau^l(t)$. Therefore
	\begin{eqnarray*}
	P_{t,x,y}(\ell(t)=l)& = & E_{t,x,y}[P_{t,x,y}(\tau^l(t)<\tau^{-l}(t)|\tau^l(t))]\\
	& = &\Int_0^\infty (1-F_{\tau^{-l}}(s|x,y))f_{\tau^l}(s|x^l,y^l) ds\\
	& = &\Int_0^\infty\displaystyle\prod_{m\neq l}(1-F_{\tau^m}(s|x,y))f_{\tau^l}(s|x^l,y^l)ds.
	\end{eqnarray*}
	Now using, \eqref{eqf} and \eqref{eqft}, one get
	\begin{eqnarray}\label{eql}
	P_{t,x,y}(\ell(t)=l)=\Int_0^\infty\displaystyle\prod_{m\neq l}\frac{1-F^m(s+y^m|x^m)}{1-F^m(y^m|x^m)}\frac{f^l(s+y^l|x^l)}{1-F^l(y^l|x^l)}ds.
	\end{eqnarray}
	This completes the proof of (i).
	
	\noi (ii) Also, from the definition of $F_{\tau^l|l}(v|x,y)$ we have,
	\begin{eqnarray*}
	F_{\tau^l|l}(v|x,y)&=&P_{t,x,y}(\tau^l(t)\leq v|\ell(t)=l)\\
	& = & \frac{P_{t,x,y}(\tau^l(t)\leq v,\ell(t)=l)}{P_{t,x,y}(\ell(t)=l)}.
	\end{eqnarray*}
	Again to compute $P_{t,x,y}(\tau^l(t)\leq v,\ell(t)=l)$ we use conditioning on $\tau^l(t)$, therefore
	\begin{eqnarray}\label{ptxytault}
	\no  P_{t,x,y}(\tau^l(t)\leq v,\ell(t)=l)&=& E_{t,x,y}[P_{t,x,y}(\tau^{-l}(t)>\tau^l(t),\tau^l(t)\leq v|\tau^l(t))]\\
	\no & = & \Int_0^v P_{t,x,y}(\tau^{-l}(t)>\tau^l(t)|\tau^l(t)=s) f_{\tau^l}(s|x^l,y^l)ds\\
	\no & = & \Int_0^v(1-P_{t,x,y}(\tau^{-l}(t)\leq s))f_{\tau^l}(s|x^l,y^l)ds\\
	& = & \Int_0^v\displaystyle\prod_{m\neq l}(1-F_{\tau^m}(s|x,y))f_{\tau^l}(s|x^l,y^l)ds .
	\end{eqnarray}
	Now substituting \eqref{eqf},\eqref{eql} in \eqref{ptxytault} we have,
	\begin{eqnarray}\label{Ftaullv}
	F_{\tau^l|l}(v|x,y)= \frac{\Int_0^v\displaystyle\prod_{m\neq l}(1-F^m(s+y^m|x^m))f^l(s+y^l|x^l)ds}{\Int_0^\infty\displaystyle\prod_{m\neq l}(1-F^m(s+y^m|x^m))f^l(s+y^l|x^l)ds}.
	\end{eqnarray}
	Since $\displaystyle\prod_{m\neq l}(1-F^m(s+y^m|x^m))f^l(s+y^l|x^l)$ is $C^1$ for all $s\in [0,T]$, by the fundamental Theorem of calculas we can conclude that $F_{\tau^l|l}(v|x,y)$
	is twice differentiable for all $v$.
	\qed \\
	\begin{eqnarray}\label{exftaul}
	f_{\tau^l|l}(v|x,y)=\frac{\displaystyle\prod_{m\neq l}(1-F^m(v+y^m|x^m))f^l(v+y^l|x^l)}{\Int_0^\infty\displaystyle\prod_{m\neq l}(1-F^m(s+y^m|x^m))f^l(s+y^l|x^l)ds},
	\end{eqnarray}
	is differentiable with respect to $v$.

	\noi The proof of (iii) is straightforward.

	\textbf{Proof of Lemma \ref{regularity}:}
	\noi(i) We will show that, $P_{t,x,y}(\ell(t)=l)$ and $F_{\tau^l|l}(T-t|x,y)$ is in dom$(D_{t,y})$. Consider a function
	$\digamma^l_v(x,y):=\int_0^v\prod_{m\neq l}(1-F^m(s+y^m|x^m))f^l(s+y^l|x^l)ds$ and $\digamma^l_\infty(x,y):=\displaystyle\lim_{v\rightarrow\infty}\digamma^l_v
	(x,y)$. Then the function $\digamma^l_v(x,y)$ is continuously differentiable with respect to $\digamma^l_v(x,y)$ and we denote this derivative by $\digamma^{l'}_v(x,y)$. Therefore,
	\begin{equation}\label{digprim}
	\digamma^{l'}_v(x,y):=\prod_{m\neq l}(1-F^m(v+y^m|x^m))f^l(v+y^l|x^l).
	\end{equation}
	Since $\digamma^l_v(x,y)$ is not depending upon $t$, we check the differentiability in $y$. To this end we first show the existence of the following limit
	\begin{eqnarray*}
	\lim_{\vp\to 0}&\frac{1}{\vp}\Big  [\int_{0}^{v}\prod_{m\neq l}(1-F^m(s+y^m+\vp|x^m))f^l(s+y^l+\vp|x^l)ds\\
	&-\int_{0}^{v}\prod_{m\neq l}(1-F^m(s+y^m|x^m))f^l(s+y^l|x^l)ds\Big].
	\end{eqnarray*}
	By a suitable substitution of variable, the expression in the above limit is
	\begin{eqnarray*}
	&\frac{1}{\vp}\Big  [\int_{v}^{v+\vp}\prod_{m\neq l}(1-F^m(s+y^m|x^m))f^l(s+y^l+\vp|x^l)ds\\
	&-\int_{0}^{\vp}\prod_{m\neq l}(1-F^m(s+y^m|x^m))f^l(s+y^l|x^l)ds\Big].
	\end{eqnarray*}
	The above expression converges to $\digamma^{l'}_v(x,y)-\digamma^{l'}_0(x,y)$ as $\vp\rightarrow 0$ and the limit is continuous in $y$. Thus $$D_{t,y}\digamma^l_v(x,y)=\digamma^{l'}_v(x,y)-\digamma^{l'}_0(x,y).$$ If $v$ is a differentiable function of $t$, then $$D_{t,y}\digamma^l_v(x,y)=\digamma^{l'}_v(x,y)\left(1+\frac{\partial v}{\partial t}\right)-\digamma^{l'}_0(x,y).$$
	Hence
	\begin{equation}\label{eq2}
	D_{t,y}\digamma_v^l(x,y)=
	\begin{cases}
	\digamma^{l'}_v(x,y)\left(1+\frac{\partial v}{\partial t}\right)-\digamma^{l'}_0(x,y) & 0<v<\infty \\
	-\digamma^{l'}_0(x,y) & v=\infty.
	\end{cases}
	\end{equation}
	Since $$D_{t,y}\prod_m(1-F^m(v+y^m|x^m))=-\sum_rf^r(y^r|x^r)\prod_{m\neq r}(1-F^m(y^m|x^m)$$
	Hence $P_{t,x,y}(\ell(t)=l)=\frac{\digamma^l_\infty(x,y)}{\prod_m(1-F^m(y^m|x^m))}$ and $F_{\tau^l|l}(T-t|x,y)=\frac{\digamma^l_{T-t}(x,y)}{\digamma^l_\infty(x,y)}$. Hence $P_{t,x,y}(\ell(t)=l)$ and $F_{\tau^l|l}(T-t|x,y)$ are in the domain of $D_{t,y}$. Now operating $D_{t,y}$ on $P_{t,x,y}(\ell(t)=l)$  and using \eqref{eqft}, \eqref{digprim} we have
	\begin{eqnarray*}
	D_{t,y}P_{t,x,y}(\ell(t)=l)&=&\frac{D_{t,y}\digamma_\infty(x,y)}{\prod_m(1-F^m(v+y^m|x^m))}\\
	&&+\frac{\digamma_\infty(x,y)\times\sum_rf^r(y^r|x^r)\prod_{m\neq r}(1-F^m(y^m|x^m))}{(\prod_m(1-F^m(v+y^m|x^m)))^2}\\
	&=&-\frac{\digamma'_0(x,y)}{\prod_m(1-F^m(v+y^m|x^m))}+\sum_{r=0}^{n}\frac{f^r(y^r|x^r)}{(1-F^r(v+y^r|x^r))}P_{t,x,y}(\ell(t)=l)\\
	&=& \sum_{r=0}^{n} f_{\tau^r}(0|x^r,y^r)P_{t,x,y}(\ell(t)=l)-f_{\tau^l}(0|x^l,y^l).
	\end{eqnarray*}
	Operating $D_{t,y}$ on $F_{\tau^l|l}(T-t|x,y)$
	\begin{eqnarray*}
	D_{t,y}F_{\tau^l|l}(T-t|x,y)&=&\frac{D_{t,y}\digamma_{T-t}(x,y)}{\digamma_\infty(x,y)}-\frac{\digamma_{T-t}(x,y)D_{t,y}\digamma_\infty(x,y)}{\digamma^2_\infty(x,y)}\\
	&=& -\frac{\digamma'_0(x,y)}{\digamma_\infty(x,y)}+\frac{\digamma_{T-t}(x,y)\digamma'_0(x,y)}{\digamma^2_\infty(x,y)}\\
	&=& f_{\tau^l|l}(0|x,y)(F_{\tau^l|l}(T-t|x,y)-1).
	\end{eqnarray*}
	
	(ii) Using assumption (A1)(ii), we can conclude that $f_{\tau^l|l}(t|x,y)$ is in dom($D_{t,y}$).

	(iii) \noi From \eqref{covsigma}, we get that $\Sigma^{-1}$ exists for all $v>0$ and is differentiable in $t$ and $v$. Therefore $\alpha(\vs;t,s,x,v)$ defined in \eqref{expalpha} is differentiable in $t$ and $v$. Taking logarithm on both the sides of \eqref{expalpha}, we have
	\begin{equation}\label{logexpalpha}
	\ln\alpha(\vs;t,s,x,v)=-\ln\left(\frac{1}{\sqrt{(2\pi)^n}\vs_1\vs_2\ldots\vs_n}\right)
	-\frac{1}{2}\ln|\Sigma|-\frac{1}{2}\sum_{ll'}\Sigma^{-1}_{ll'}(z^l-\bar{z}^l)(z^{l'}-\bar{z}^{l'}).
	\end{equation}
	Now taking derivative on both the sides of \eqref{logexpalpha} with respect to $t$ and using Jacobi's formula.
	\begin{align}\label{alphat}
\no	\alpha_t&=\alpha\left(-\frac{1}{2}\frac{|\Sigma|}{
		|\Sigma|}tr\left(\Sigma^{-1}\frac{\partial\Sigma}{\partial t}\right)-\frac{1}{2}\sum_{ll'}\Sigma^{-1}_{ll't}(z^l-\bar{z}^l)(z^{l'}-\bar{z}^{l'})+\frac{1}{2}\sum_{ll'}\Sigma^{-1}_{ll'}\bar{z}_t^l(z^{l'}-\bar{z}^{l'})+\frac{1}{2}\sum_{ll'}\Sigma^{-1}_{ll'}(z^l-\bar{z}^l)\bar{z}^{l'}_t\right)\\
	&=\alpha\left(-\frac{1}{2}tr\left(\Sigma^{-1}\frac{\partial\Sigma}{\partial t}\right)-\frac{1}{2}\sum_{ll'}\Sigma^{-1}_{ll't}(z^l-\bar{z}^l)(z^{l'}-\bar{z}^{l'})+\frac{1}{2}\sum_{ll'}\Sigma^{-1}_{ll'}\bar{z}_t^l(z^{l'}-\bar{z}^{l'})+\frac{1}{2}\sum_{ll'}\Sigma^{-1}_{ll'}(z^l-\bar{z}^l)\bar{z}^{l'}_t\right).
	\end{align}
	Similarly
	\begin{equation}\label{alphav}
	\alpha_v=\alpha\left(-\frac{1}{2}tr\left(\Sigma^{-1}\frac{\partial\Sigma}{\partial v}\right)-\frac{1}{2}\sum_{ll'}\Sigma^{-1}_{ll'v}(z^l-\bar{z}^l)(z^{l'}-\bar{z}^{l'})+\frac{1}{2}\sum_{ll'}\Sigma^{-1}_{ll'}\bar{z}_v^l(z^{l'}-\bar{z}^{l'})+\frac{1}{2}\sum_{ll'}\Sigma^{-1}_{ll'}(z^l-\bar{z}^l)\bar{z}^{l'}_v\right),
	\end{equation}
	where $\Sigma^{-1}_{ll't}:=\frac{\pa \Sigma^{-1}_{ll}}{\pa t}$ and $\Sigma^{-1}_{ll'v}:=\frac{\pa \Sigma^{-1}_{ll}}{\pa v}$ . In similar manner one can show $\alpha$ is infinite times continuously differentiable in $s$. \qed	\section*{Acknowledgement} The second author acknowledges Mrinal K. Ghosh for some useful discussion.

\end{document}